\documentclass[12pt]{extarticle}
\def\version{23 January 2026}

\nonstopmode
\usepackage{tocloft}
\usepackage{enumitem}
\usepackage{booktabs}
\usepackage{ifpdf}
\usepackage[normalem]{ulem}
\newcommand{\coupling}{\mathrm{g}}
\usepackage{multirow}
\usepackage{arydshln}
\usepackage{xr}
\usepackage{latexsym}
\usepackage{bm}
\usepackage{upgreek}

\usepackage{mathrsfs}
\usepackage{times}
\usepackage{amsmath}
\usepackage{amsthm}
\usepackage{amssymb}
\usepackage{mathtools}
\usepackage{epsfig}
\usepackage{graphicx}
\usepackage{subcaption}

\usepackage{epigraph}
\usepackage{comment}

\usepackage[table, dvipsnames]{xcolor}
\usepackage{color}
\definecolor{MyDarkBlue}{rgb}{0,0.08,0.45}

\usepackage[backref=none]{hyperref}
\hypersetup{pdfborder={0 0 0},
  colorlinks,
  urlcolor={MyDarkBlue},
  linkcolor={MyDarkBlue},
  citecolor={MyDarkBlue},
  breaklinks=true}
\usepackage{doi}
\usepackage{orcidlink}

\providecommand{\eprint}[1]{}
\renewcommand{\eprint}[1]{arXiv:\href{http://arxiv.org/abs/#1}{#1}}

\newcommand{\jj}{\mathrm{i}}

\DeclareOldFontCommand{\brianup}{\upshape}{\mathrm}

\DeclareSymbolFont{EUR}{U}{eur}{m}{n}
\SetSymbolFont{EUR}{bold}{U}{eur}{b}{n}
\DeclareSymbolFontAlphabet{\eur}{EUR}

\DeclareSymbolFont{EUB}{U}{eur}{b}{n}
\SetSymbolFont{EUB}{bold}{U}{eur}{b}{n}
\DeclareSymbolFontAlphabet{\eub}{EUB}

\DeclareSymbolFont{AMSb}{U}{msb}{m}{n}
\DeclareSymbolFontAlphabet{\mathbb}{AMSb}

\newcommand{\p}{\partial}

\newcommand{\at}[1]{\vert\sb{\sb{#1}}}

\def\R{\mathbb{R}}
\newcommand{\C}{\mathbb{C}}
\newcommand{\N}{\mathbb{N}}
\newcommand{\abs}[1]{\vert #1 \vert}
\newcommand{\norm}[1]{\Vert #1 \Vert}

\DeclareMathSymbol{\varGamma}{\mathord}{letters}{"00}
\DeclareMathSymbol{\varDelta}{\mathord}{letters}{"01}
\DeclareMathSymbol{\varTheta}{\mathord}{letters}{"02}
\DeclareMathSymbol{\varLambda}{\mathord}{letters}{"03}
\DeclareMathSymbol{\varXi}{\mathord}{letters}{"04}
\DeclareMathSymbol{\varPi}{\mathord}{letters}{"05}
\DeclareMathSymbol{\varSigma}{\mathord}{letters}{"06}
\DeclareMathSymbol{\varUpsilon}{\mathord}{letters}{"07}
\DeclareMathSymbol{\varPhi}{\mathord}{letters}{"08}
\DeclareMathSymbol{\varPsi}{\mathord}{letters}{"09}
\DeclareMathSymbol{\varOmega}{\mathord}{letters}{"0A}

\theoremstyle{plain}

\newtheorem{lemma}{Lemma}[section]

\newcommand\xqed[1]{%
	\leavevmode\unskip\penalty9999 \hbox{}\nobreak\hfill\quad\hbox{#1}%
}
\newcommand\remarkend{\xqed{$\triangle$}}
\makeatletter
	\def\@endtheorem{\remarkend\endtrivlist\@endpefalse }
\makeatother

\theoremstyle{remark}
\newtheorem{remark}[lemma]{Remark}
\newtheorem*{remark*}{Remark}
\makeatletter
	\def\@endtheorem{\endtrivlist\@endpefalse }
\makeatother

\renewcommand{\theequation}{\thesection.\arabic{equation}}

\makeatletter\@addtoreset{equation}{section}
\makeatletter\@addtoreset{theorem}{section}
\makeatother

\renewcommand{\Re}{\mathop{\rm{R\hskip -1pt e}}\nolimits}
\renewcommand{\Im}{\mathop{\rm{I\hskip -1pt m}}\nolimits}

\providecommand{\keywords}[1]
{
\noindent
\small \textbf{\textit{Keywords:}} #1
}

\begin{document}
\renewcommand{\theequation}{\thesection.\arabic{equation}}
\newcommand{\sect}[1]{\setcounter{equation}{0}\section{#1}}

\title{Numerical study
of solitary waves\\in
Dirac--Klein--Gordon system
}

\author{
Andrew Comech\,\orcidlink{0000-0003-2112-9359}
\\
{\small\it
Mathematics Department,
Texas A\&M University, College Station, TX 77843, USA}
\\[1ex]
Julien Ricaud\,\orcidlink{0000-0002-6777-7469}
\\
\small\it
Departamento de F\'{i}sica Matem\'{a}tica,
Instituto de Investigaciones en Matem\'aticas Aplicadas y en Sistemas,
\\
\small\it
Universidad Nacional Aut\'onoma de M\'exico, CP 04510, Ciudad de M\'exico, M\'exico
\\[1ex]
Marco Roque
\\
\small\it
Caver College of Aviation, 
Science, and Nursing,
Henderson State University, Arkadelphia, AR 71923
\\
\small\it
Arkansas State University,
State University, AR 72467
}

\date{\version}

\maketitle

\keywords{
Dirac--Klein--Gordon system,
Nonlinear Dirac equation,
solitary waves,
virial identity,
nested shooting method,
iterative method,
stability of solitary waves.
}

\begin{abstract}
We use numerics to construct solitary waves
$\phi_\omega(x) e^{-\jj\omega t}$
in Dirac--Klein--Gordon 
(in one and three spatial dimensions)
and study the dependence of energy and charge of $\omega$.
To construct solitary waves,
we use two different procedures:
the iterative method
and the nested shooting method.
We also consider the case of massless scalar
field where we show that the
standard shooting method becomes available.
We use the virial identities
to control the error of simulations.
We discuss possible implications
for the stability
of solitary waves.
\end{abstract}

\setlength{\cftbeforesecskip}{5pt}
\tableofcontents

\section{Introduction}

While the nonlinear Dirac equation
was written down by D.\,Ivanenko back in
1938 \cite{jetp.8.260}
and the nonlinear Klein--Gordon equation
by L.\,Schiff
in 1951 \cite{PhysRev.84.1,PhysRev.84.10},
the rigorous mathematical study
of nonlinear field equations
started forming in early sixties, already after
the triumph of Quantum Electrodynamics.
The Cauchy problem for the nonlinear wave equation
appears in \cite{jorgens1961anfangswertproblem,segal1963global}.
This research was followed by the study of
nonlinear scattering
(stability of zero solution)
\cite{MR0217453,MR0233062,MR0303097},
existence of solitary waves
\cite{strauss1977existence,berestycki1983nonlinear,cazenave1986existence,esteban1996stationary},
and then linear and orbital
stability of solitary waves
\cite{zakharov-1967,kolokolov-1973,cazenave1982orbital,grillakis1987stability} followed by the ongoing research on asymptotic stability
of solitary waves.

The results for systems
appear with a delay,
particularly so
because of the difficulties caused by
nonlocal type of self-interaction.
The Dirac--Klein--Gordon system
(DKG) -- the focus of the present article --
describes
fermions interacting with the
scalar boson field
via Yukawa coupling.
The ``Higgs portal''
theories suggest that this model
may be relevant for the description
of the Dark Matter
(see \cite{comech2025stable} and the references therein).
The relation
between the masses
of fermions and
bosons ($m$ and $M$, respectively)
in these theories
is presently not clear;
one possibility is to have light
fermions
interacting with the
Higgs field,
$m\ll M=125\,\mathrm{GeV}$.
On the other side of the spectrum,
$m\gg M$,
heavy fermions
could interact via exchange of
axions
in models based on axial currents
$\bar\psi\gamma^5\gamma^\mu\psi$
(where the Dirac conjugate is defined by
$\bar\psi=\psi^*\gamma^0$)
or via exchange of 
a
very light or massless spinless neutral boson
called majoron,
arising in models with spontaneous lepton number breaking,
based on Yukawa-type coupling terms
$\bar\psi\psi \varphi$
or $\bar\psi\gamma^5\psi \varphi$
\cite{lessa2007revisiting}.

DKG, Dirac--Maxwell (DM),
and similar systems
admit localized modes
(see \cite{esteban1996stationary,abenda1998solitary,MR2593110}),
which may or may not have physical relevance
in the description of nonperturbative effects of quantum field theories;
the stability
of localized solutions
is crucial for their possible applicability
in the description of certain physical phenomena.
Remarkably, the spectral stability of standing waves
$\phi_\omega(x)e^{-\jj\omega t}$
-- that is, the absence of linear instability --
is available
in the nonlinear Dirac equation (NLD):
see
\cite{berkolaiko2012spectral}
for the numeric results on spectral stability
in 1D
(corroborated numerically and
analytically in
\cite{arxiv:1707.01946,lakoba2018numerical,aldunate2023results})
and
\cite{PhysRevLett.116.214101,boussaid2019spectral}
for numeric and analytic results
on spectral stability in higher dimensions.
Orbital stability of solitary waves
has been shown
in a completely integrable massive Thirring model \cite{pelinovsky2014orbital,MR3462129}.
See also the monograph
\cite{boussaid2019nonlinear} for more details
and references.

The articles
\cite{comech2013polarons,comech2025stable}
suggest spectral stability of
some of the solitary waves in
DM and DKG systems.
To set up the groundwork for the research
in these two systems,
it is convenient to have some basic properties of
the corresponding solitary waves.
The problem is that solitary waves
in Dirac-based models
are not readily available:
there are no explicit formulas known
(except for 1D NLD),
the proof of their existence is indirect
\cite{esteban1996stationary},
and the asymptotic behavior is only available
in the nonrelativistic limit
(see
\cite{boussaid2017nonrelativistic}
for NLD for $\omega$ near $m$
and \cite{comech2018small}
for DM
for $\omega$ near $-m$).
This is why we have to rely on the numerics
to know the behavior of charge and energy
of solitary waves $\phi(x,\omega)e^{-\jj\omega t}$
as functions of frequency,
to be able to make at least some predictions
regarding stability:
the stability of small amplitude solitary waves
is inherited from the nonrelativistic limit of the model
(see, e.g., \cite{comech2014linear,boussaid2019spectral})
and then critical points of
the charge $Q(\omega)$
and zero values of $E(\omega)$
indicate the collision of eigenvalues of the
linearization operators at the origin and hence
possibly the border of the stability region
\cite{berkolaiko2015vakhitov}.

In the present article,
we give a numerical construction
of solitary waves
in DKG with Yukawa coupling.
Since one needs to specify two parameters
for the construction of DKG solitary waves,
the shooting method needs to be modified.
We develop the nested shooting method
which allows to adjust two parameters.
We also use the iterative procedure
which converges to DKG solitary waves.
In the massless case
(Dirac equation coupled to the wave equation),
we can also use the iterative approach
(albeit with a poor precision);
besides, the standard,
one-parameter shooting method
becomes available,
yielding results with much higher
accuracy.
We show that indeed both methods
are in agreement.
While we cannot yet produce
results on spectral stability,
we will already be able to make
some preliminary conclusions
about the location of the stability region
based on the dependence of the energy
of $\omega$.
We point out that while
there are no zero spin massless particles
in the standard model as of now
(the only massless particle
which is presently known
in High Energy Physics
is photon, a spin-$1$ particle),
the conclusion of the present paper is that
the massive spinor field coupled to massless
scalar field
forms localized states,
and we expect these states
to be stable
at least for
$\omega\lesssim m$
(``the nonrelativistic limit'').\footnote{
By ``$\omega\lesssim m$'' we mean that
$\omega\in(\omega_1,m)$
with some $\omega_1<m$.
}

Let us mention that the nonlinear Dirac and Dirac--Klein--Gordon equations
admit bi-frequency solitary wave solutions
of the form
$\phi(x)e^{-\jj\omega t} + \chi(x)e^{\jj\omega t}$
\cite{boussaid2018spectral}.
By
\cite{boussaid2024spectral,comech2025stable},
these bi-frequency solitary waves
have spectral stability properties
similar to those of standard, one-frequency solitary waves,
although it is only bi-frequency solitary waves that could be asymptotically stable.
Since the bi-frequency solitary waves
are readily produced
from one-frequency ones
(see Remark~\ref{remark-bi}
below),
in the present article we only
focus on the construction of
one-frequency solitary waves.

The DKG system is introduced in
Section~\ref{sect-dkg}.
The numerical setup for the case
of massive scalar field
is in Section~\ref{sect-dkg-numerics};
the massless case is considered in
Section~\ref{sect-dkg-numerics-massless}.
The stability properties
are discussed in Section~\ref{sect-stability}.
Relevant details about NLD
are presented in Appendix~\ref{sect-nld}.

\section{Dirac--Klein--Gordon system}
\label{sect-dkg}

The Dirac--Klein--Gordon system
describes
the spinor field
$\psi\in\C^N$
coupled to the real-valued
Klein--Gordon field $\varphi\in\R$:
\begin{equation}\label{dkg}
    \begin{cases}
        \jj\p\sb{t}\psi
        =D_0\psi+(m-\coupling \varphi)\beta\psi\,,
        \qquad
        \psi(t,x)\in\C^N,
        \ \ x\in\R^n,
        \\
        (\p_t^2-\Delta+M^2)
        \varphi
        =\psi^*\beta\psi\,,
        \qquad
        \varphi(t,x)\in\R\,.
    \end{cases}
\end{equation}
Above, $\coupling>0$ is the coupling constant,
$m>0$ is the mass of the spinor field,
and $M\ge 0$ is the mass of the boson field
(we note that we allow $M=0$;
in this case, formally one may call
the system
\eqref{dkg} \emph{the Dirac\,--\,wave equation};
$D_0 = -\jj\bm\alpha\cdot\nabla$
is the Dirac operator.
The Dirac matrices
$\alpha^j$ ($1\le j\le n$)
and $\beta$
are self-adjoint
and satisfy
\begin{equation}\label{anti-alpha}
\alpha^j\alpha^k
+\alpha^k\alpha^j=2\delta_{j k} I_N\,,
\qquad
\alpha^j\beta
+\beta\alpha^j=0\,,
\qquad
1\le j,\,k\le n\,;
\qquad
\beta^2=I_N\,;
\end{equation}
We focus on the 3D case,
taking $n=3$ and $N=4$
(see also Remark~\ref{remark-nd-1}
below),
taking the Dirac matrices in the standard form
\[
    \alpha^j
    =
    \begin{bmatrix} 0&\sigma_j
    \\ \sigma_j&0\end{bmatrix},
    \quad
    \beta
    =
    \begin{bmatrix} I_2&0 \\ 0&-I_2\end{bmatrix};
    \qquad
    \sigma_1=\begin{bmatrix}0&1\\1&0\end{bmatrix},
    \quad
    \sigma_2=\begin{bmatrix}0&-\jj\\\jj&0\end{bmatrix},
    \quad
    \sigma_3=\begin{bmatrix}1&0\\0&-1\end{bmatrix}.
\]

\begin{remark}
\label{remark-coupling}
One can see that the constant $\coupling$ in~\eqref{dkg}
can be assumed to be equal to $1$.
Indeed, 
$\tilde\psi=\coupling^{1/2}\psi$ and
$\tilde \varphi=\coupling \varphi$ satisfy
the same system \eqref{dkg}
but with $\coupling=1$.
So, to consider physically
different cases, it is
enough to take e.g. $m=1$
(which could be achieved by rescaling $x$)
and only consider different
values of $M\ge 0$.
All other cases are obtained by scaling and could be ignored.
\end{remark}

\begin{remark}\label{remark-close}
Formally, the system \eqref{dkg}
turns into the nonlinear Dirac equation
(see \eqref{nld} below)
in the limit
$\coupling=M^2\to+\infty$.
\end{remark}

In the 3D case,
by \cite{esteban1996stationary},
for any $\omega\in(0,m)$,
there exist solitary wave solutions
\[
    (\psi(t,x),\,\varphi(t,x))
    =\bigl(\phi(x,\omega)e^{-\jj\omega t},\,h(x,\omega)\bigr)
\]
to
\eqref{dkg},
with
$\phi\in\C^4$ and $h\in\R$
satisfying
\begin{equation}\label{dkg-stationary}
    \begin{cases}
    \omega\phi
    =D_0\phi+(m-\coupling h)\beta\phi\,,
    \\
    (-\Delta+M^2)h= \phi^*\beta\phi\,,
    \end{cases}
\end{equation}
with $\phi(x,\omega)$
in the form of the first Wakano Ansatz~\cite{wakano-1966},
\begin{equation}\label{sw-1}
    \phi(x,\omega)
    =
    \begin{bmatrix}
    v(r,\omega)\begin{bmatrix}1\\0\end{bmatrix}
    \\
    \jj u(r,\omega)
    \begin{bmatrix}\cos\theta\\e^{\jj\upvarphi}\sin\theta\end{bmatrix}
    \end{bmatrix}
    ,
\end{equation}
where $(r,\theta,\upvarphi)$
are the spherical coordinates.

\begin{remark}
\label{remark-bi}
By
\cite{boussaid2018spectral},
the existence in NLD or in DKG
of a solitary wave
with $\phi$ of the form
\eqref{sw-1}
implies the existence of
bi-frequency solitary wave solutions  
of the form
\begin{align}
\begin{bmatrix}
v(r,\omega)\bm\xi
\\
\jj u(r,\omega)\sigma_r\bm\xi
\end{bmatrix}
e^{-\jj\omega t}
+
\begin{bmatrix}
-\jj u(r,\omega)\sigma_r\bm\eta
\\
v(r,\omega)\bm\eta
\end{bmatrix}
e^{\jj\omega t},
\qquad
\bm\xi,\,\bm\eta\in\C^2,
\quad
\norm{\bm\xi}^2-\norm{\bm\eta}^2=1,
\end{align}
where
$
\sigma_r
:=
\abs{x}^{-1}x\cdot\bm\sigma
=
\begin{bmatrix}
\cos\theta&e^{-\jj\upvarphi}\sin\theta
\\
e^{\jj\upvarphi}\sin\theta
&-\cos\theta
\end{bmatrix}$
the Pauli matrix in spherical coordinates.
The stability properties
of bi-frequency solitary waves
seem related (but not identical)
to those of standard,
one-frequency solitary waves.
For more details,
see \cite{comech2025stable}.
\end{remark}

\begin{remark}\label{remark-nd-1}
For the generalization of the Wakano
Ansatz \eqref{sw-1}
to any dimension $n\ge 1$
see, e.g.,
\cite[Chapter XII]{boussaid2019nonlinear}.
\end{remark}

\begin{lemma}
Let $n\in\N$.
There are no solitary wave solutions
to \eqref{dkg}
with $\omega\in(-m,m)$,
$\phi\in L^2(\R^n,\C^N)$,
$h\in L^\infty(\R^n)$,
such that
$\phi^*\beta\phi\le 0$
everywhere in $\R^n$.
\end{lemma}
\begin{proof}
Denote the upper and the lower
parts of the solitary wave profile
$\phi$ by $\phi_P,\,\phi_A\in\C^{N/2}$.
Then one has:
\[
    \omega
    \begin{bmatrix}\phi_P\\\phi_A\end{bmatrix}
    =
    \begin{bmatrix}0&
    -\jj\bm\sigma\cdot\nabla
    \\
    -\jj\bm\sigma\cdot\nabla
    &0\end{bmatrix}
    \begin{bmatrix}\phi_P\\\phi_A\end{bmatrix}
    +
    (m-\coupling h)
    \begin{bmatrix}\phi_P\\\phi_A\end{bmatrix},
    \qquad
    (-\Delta+M^2) h=\phi^*\beta\phi\,.
\]
Let
$\phi\in L^2(\R^n,\C^N)$,
$h\in L^\infty(\R^n,\R)$;
it follows that
$\phi\in H^2(\R^n,\C^N)$.
If we assume that
$\phi^*\beta\phi\le 0$ everywhere
in $\R^n$,
then also
$h=
(-\Delta+M^2)^{-1}
\phi^*\beta\phi\le 0$
(due to the positivity of the integral kernel
which corresponds to the
Bessel potential $(-\Delta+M^2)^{-1}$).
Writing
\[
    (\omega-m+\coupling h)\phi_P
    =-\jj\bm\sigma\cdot\nabla\phi_A,
    \qquad
    (\omega+m-\coupling h)\phi_A
    =-\jj\bm\sigma\cdot\nabla\phi_P\,,
\]
we conclude that
\[
    (\omega-m+\coupling h)\phi_P
    =-\jj\bm\sigma\cdot\nabla
    \frac{1}{\omega+m-\coupling h}
    (-\jj\bm\sigma\cdot\nabla)\phi_P\,.
\]
Coupling this with $\phi_P$ and
integrating and taking
into account that
$\omega+m-\coupling h\ge\omega+m>0$,
$\omega-m+\coupling h\le\omega-m<0$,
we see that the left-hand side
is strictly negative while the
right-hand is strictly positive, thus
a contradiction.
\end{proof}

Our goal is to compute numerically
the charge and the energy
corresponding to the solitary waves
$\phi(x,\omega)e^{-\jj\omega t}$,
with the corresponding boson field
$h(x,\omega)$:
\begin{align}\label{def-Q}
    Q(\phi)
    &=\int_{\R^n}\phi^*\phi\,,
    \\
    E(\phi,h)
    &=\int_{\R^n}
    \Bigl(\phi^*D_0\phi
    +m \phi^*\beta\phi
    -\coupling h\,\phi^*\beta\phi
    +
    \frac{\coupling}{2}\abs{\nabla h}^2
    +\frac{\coupling M^2}{2}h^2
    \Bigr)
    \nonumber
    \\
    &=K(\phi)+N(\phi)+V(\phi,h)
    +T(h)+W(h)\,,
\end{align}
where
\begin{equation}
\label{def-kn}
    K(\phi)
    =\int_{\R^n}
    \phi^* D_0\phi\,,
    \qquad
    N(\phi)
    =m\int_{\R^n} \phi^*\beta\phi\,,
\end{equation}
\begin{equation}\label{def-tvw}
    T(h)=\frac{\coupling}{2}
    \int_{\R^n} \abs{\nabla h}^2\,,
    \qquad
    V(\phi,h)=-\coupling\int_{\R^n} h\,\phi^*\beta\phi\,,
    \qquad
    W(h)=\frac{\coupling M^2}{2}\int_{\R^n} h^2\,.
\end{equation}
There are 
immediate identities
that follow from
\eqref{dkg-stationary}
by multiplication with
$\phi^*$ and $h$,
respectively, and integration:
\begin{equation}\label{short-2}
    \omega Q(\phi)
    =K(\phi)+N(\phi)+V(\phi,h)\,,
    \qquad
    2T(h)+2W(h)=-V(\phi,h)\,.
\end{equation}
The relations
\eqref{short-2}
allow us
to derive the
expression
\begin{equation}\label{dkg-energy}
    E(\phi,h) = K(\phi)+N(\phi)+V(\phi,h)+T(h)+W(h) =\omega Q(\phi)+T(h)+W(h)
\end{equation}
for the total energy and alternatively also
\begin{equation}
\label{dkg-energy-no-derivatives}
    E(\phi,h) =\omega Q(\phi)-\frac{1}{2}V(\phi,h)\,,
\end{equation}
which we will use in the numerics to compute $E$.

The system
\eqref{dkg-stationary}
can be written as
the following relations
for variational derivatives:
\begin{align}
    \label{dkg-variational-1}
    &
    \omega \delta_{\phi^*} Q
    =\delta_{\phi^*} K+\delta_{\phi^*} N
    +\delta_{\phi^*} V\,,
    \\[1ex]
    \label{dkg-variational-2}
    &
    \delta_h T+\delta_h W
    =-\delta_h V\,.
\end{align}
Above,
$\delta_{\phi^*}$ and $\delta_h$
are the variational derivatives
with respect to $\phi^*$ and $h$
(more rigorously, one defines
$\delta_{\phi}
=\frac 1 2(\delta_{\Re\phi}
-\jj\delta_{\Im\phi})$
and
$\delta_{\phi^*}
=\frac 1 2(\delta_{\Re\phi}
+\jj\delta_{\Im\phi})$).
Considering the dependence
of $\phi$
and $h$ of $\omega$
and
coupling
equation \eqref{dkg-variational-1}
with $\p_\omega\phi^*$,
adding to the result its complex conjugate,
that is,
the relation
\[
    \omega\langle\delta_\phi Q,\p_\omega\phi\rangle
    =
    \omega\langle\delta_\phi K,\p_\omega\phi\rangle
    +
    \omega\langle\delta_\phi N,\p_\omega\phi\rangle
    +
    \omega\langle\delta_\phi V,\p_\omega\phi\rangle\,,
\]
and then subtracting
\eqref{dkg-variational-2}
coupled with $\p_\omega h$,
one arrives at the relation
\begin{equation}\label{dkg-dedq}
    \omega \p_\omega Q
    =
    \p_\omega
    (K+N+V+T+W)
    =\p_\omega E(\omega)\,,
\end{equation}
where
$Q$, $K$, $V$, $T$, $W$,
and
$E$
evaluated at
$\phi(\omega),\,h(\omega)$
are considered as functions
of $\omega$.
We took into account that
\begin{equation}\label{vvv}
    \p_\omega V(\phi,h)
    =
    \langle \delta_h V,\p_\omega h \rangle
    +
    \langle \delta_{\phi^*} V,\p_\omega \phi^* \rangle
    +
    \langle \delta_{\phi} V,\p_\omega \phi\rangle\,.
\end{equation}

\subsubsection*{The virial identity for Dirac--Klein--Gordon system}

Denote
\[
    \phi_\lambda(x,\omega)=\phi(x/\lambda,\omega)\,,
    \qquad
    h_\lambda(x,\omega)=h(x/\lambda,\omega)\,,
    \qquad
    x\in\R^n,
    \quad\omega\in(0,m)\,,
    \quad
    \lambda>0\,.
\]
Coupling
\eqref{dkg-variational-1}
with
$\p_\lambda\phi_\lambda^*$
at $\lambda=1$
and adding to the result
its complex conjugate,
coupling
\eqref{dkg-variational-2}
with $\p_\lambda h_\lambda$
at $\lambda=1$,
and taking the difference,
one arrives at the relation
\begin{equation}
    \omega \p_\lambda\at{\lambda=1} Q(\phi_\lambda)
    =
    \p_\lambda\at{\lambda=1}
    \bigl(K(\phi_\lambda)+N(\phi_\lambda)
    +V(\phi_\lambda,h_\lambda)
    +T(h_\lambda)+W(h_\lambda)\bigr)\,.
\end{equation}
Here, we used the equivalent of \eqref{vvv}
for $\p_\lambda V(\phi_\lambda,h_\lambda)$.
Taking into account
scaling of the expressions
for $Q$, $K$, $N$, $V$, $T$,
and $W$
(cf. \eqref{nld-scaling}),
one arrives at the
virial identity
\begin{equation}\label{dkg-virial}
    \omega Q(\phi)
    =\frac{n-1}{n}K(\phi)+N(\phi)+V(\phi,h)+\frac{n-2}{n}T(h)+W(h)\,.
\end{equation}
Using the definition~\eqref{dkg-energy} of $E$,
we can write
the above virial identity as
\begin{equation}\label{wcr}
    \omega Q(\phi)
    =E(\phi,h)-\frac{1}{n}K(\phi)-\frac{2}{n}T(h)\,.
\end{equation}
Substituting $T(h)=-V(\phi,h)/2-W(h)$ from \eqref{short-2} into~\eqref{dkg-virial},
we can also write the virial identity
in the form
\begin{equation}\label{for-epsilon}
    \omega Q(\phi)
    =\frac{n-1}{n}K(\phi) + N(\phi) +\frac{n+2}{2n}V(\phi,h)+\frac{2}{n}W(h)\,.
\end{equation}
Using again
\eqref{short-2} in order to eliminate $K(\phi)$ from~\eqref{for-epsilon}, we can also write
the virial identity in the following equivalent form
(we will use this expression
for the error estimates):
\begin{equation}
\label{dkg-virial-noderivatives}
    \omega Q(\phi)
    =
    N(\phi)+\frac{4-n}{2} V(\phi,h) + 2 W(h)
\end{equation}
or, using~\eqref{dkg-energy-no-derivatives}, finally as
\begin{equation}
\label{dkg-virial-noderivatives-with-E}
    E(\phi,h) =\omega Q(\phi)-\frac{1}{2}V(\phi,h) = N(\phi) + \frac{3-n}{2} V(\phi,h) + 2 W(h)\,.
\end{equation}

\begin{lemma}
\label{lemma-en}
Nonzero solitary wave solutions
to the Dirac--Klein--Gordon equation
in $\R^n$, $n\in\N$,
with $\omega\ge 0$
satisfy
$E(\omega)>0$
and $N(\omega)>0$.
\end{lemma}

\begin{proof}
We write the first relation
in \eqref{short-2}
and \eqref{dkg-virial}
as
\[
N-\omega Q=-K-V,
\qquad
N-\omega Q=
-\frac{n-1}{n}K-V-\frac{n-2}{n}T-W.
\]
Excluding $K$, we have:
\begin{align*}
(2n-1)(N-\omega Q)
&=
-(2n-1)V
-(n-2)T-n W
=3n T+(3n-2)W;
\end{align*}
in the second equality,
we used the second relation
from \eqref{short-2}.
Now the positivity of $N$ follows from
$T>0$ and $W\ge 0$
(we note that one can have $W=0$ if
the mass $M$ of the scalar field vanishes).
Positivity of
$E(\omega)$ for $\omega\ge 0$
follows from
\eqref{dkg-energy-no-derivatives},
where $V<0$
in view of the second
relation in \eqref{short-2}.
\end{proof}

\begin{lemma}
\label{lemma-no-negative-omega}
Nonzero solitary wave solutions
to the Dirac--Klein--Gordon equation in 
$\R^n$, $n\ge 2$,
satisfy
\[
K(\omega)\ge 0;
\]
the inequality becomes strict
if
additionally
$M>0$ or $n\ge 3$.
\end{lemma}

\begin{proof}
Subtracting from \eqref{dkg-virial}
the first relation from \eqref{short-2}
yields
\[
 0=-\frac{1}{n}K(\omega)
 +\frac{n-2}{n}T(\omega)+W(\omega)\,.
\]
Since
$T>0$ and $W\ge 0$,
the above relation for $n\ge 2$ leads to
$K(\omega)\ge 0$.
If either $M>0$ (so that $W>0$)
or $n\ge 3$,
the inequality becomes strict.
\end{proof}

By
\cite{berkolaiko2015vakhitov},
the collision of eigenvalues
of the linearized operator
at the origin
is characterized by
the Vakhitov--Kolokolov
condition
$\p_\omega Q(\omega)=0$
and by the energy vanishing condition
$E(\omega)=0$;
these two indicate the jump in the
dimension of the Jordan block
corresponding to eigenvalue $\lambda=0$.
One can see on
Figure~\ref{fig-3d}
that the energy of solitary waves
remains positive for all $\omega$,
while the minimum of $E(\omega)$
-- located at approximately
$\omega=0.936m$ in the case
of the Soler model (the limit $M^2\to+\infty$)
--
for finite values of $M$
moves closer to
$\omega=m$,
suggesting larger value of the
right boundary point of the spectral stability region.

\section{Numerical construction of solitary waves}
\label{sect-dkg-numerics}

Substituting solitary waves
of the form \eqref{sw-1}
into the system
\eqref{dkg-stationary},
we see that
$v(r,\omega)$, $u(r,\omega)$,
and $h(r,\omega)$
are to satisfy the system
\begin{equation}\label{system-dkg-r}
\begin{cases}
        \p_r u
        =
        -\frac{n-1}{r}u-(m-\omega-\coupling h)v\,,
        \\
        \p_r v=
        -(m+\omega-\coupling h)u\,,
        \\
        \bigl(
        -\p_r^2-\frac{n-1}{r}\p_r
        +M^2\bigr)h=v^2-u^2\,,
    \end{cases}
    \qquad
    r>0\,.
\end{equation}

We focus on the 3D case, $n=3$,
in view of potential
applications in physics.
Let us nonetheless mention
that the 1D case is
readily accessible via both the iterative and the nested shooting methods; see Fig.~\ref{fig-1d}--\ref{fig-1-1.00} in Appendix~\ref{sect_appendix_Tables_3D},
and other dimensions via the nested shooting method.

\begin{remark}\label{remark-nd-2}
In the case $n=1$,
one considers \eqref{system-dkg-r}
with $x\in\R$ instead of $r>0$.
To arrive at this system,
one takes
$\phi(x,\omega)=\begin{bmatrix}v(x,\omega)\\u(x,\omega)\end{bmatrix}$
in place of \eqref{sw-1}
and substitutes $\alpha^1$ by $-\sigma_2$
and $\beta$ by $\sigma_3$.
For more details and for other dimensions,
we refer to
\cite[Chapters IX,\,XII]{boussaid2019nonlinear}.
\end{remark}

We obtain solitary waves, with excellent agreement between both methods, for various values of $M$ in~\eqref{system-dkg-r}.
Fig.~\ref{fig-3-1.00} and~\ref{Figure-3D-NestedshootingMethod-M=1} present for the iterative and the nested shooting methods, respectively, the profiles of the solitary waves for $M=1$ and for several values of $\omega$.
\begin{figure}[!hbt]
\ifpdf
\noindent\includegraphics[width=0.46\textwidth,height=140pt]{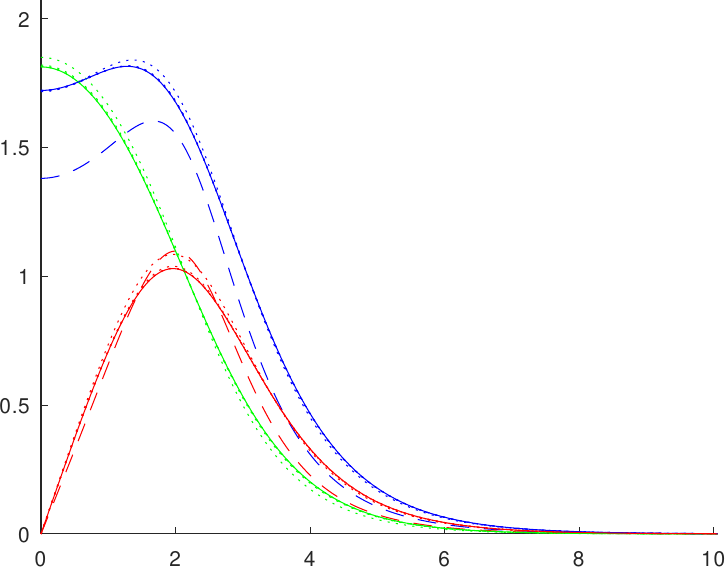}
\hfill\includegraphics[width=0.46\textwidth,height=140pt]{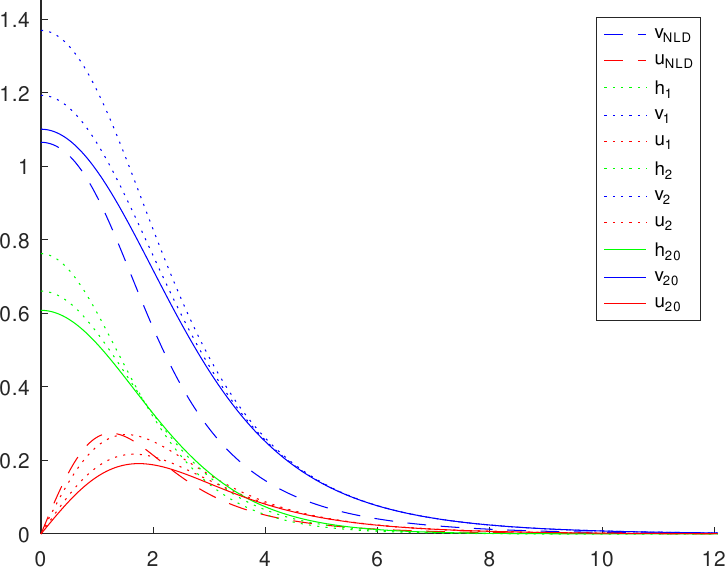}
\else
\noindent\includegraphics[width=0.46\textwidth,height=140pt]{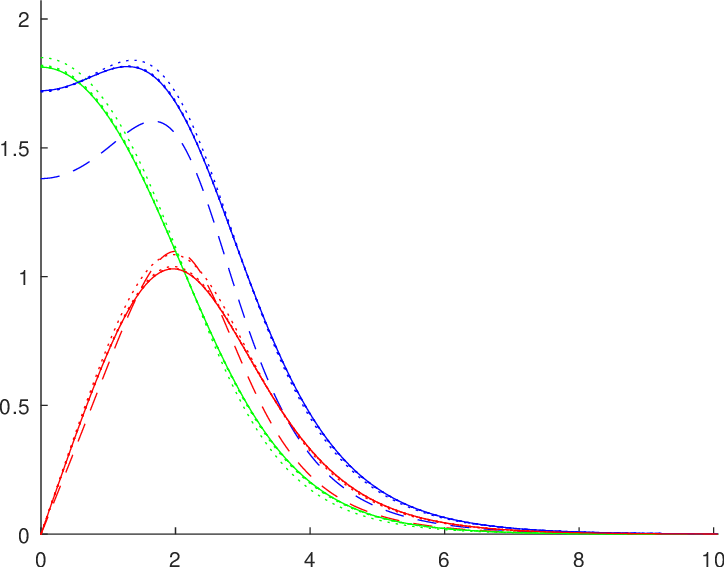}
\hfill\includegraphics[width=0.46\textwidth,height=140pt]{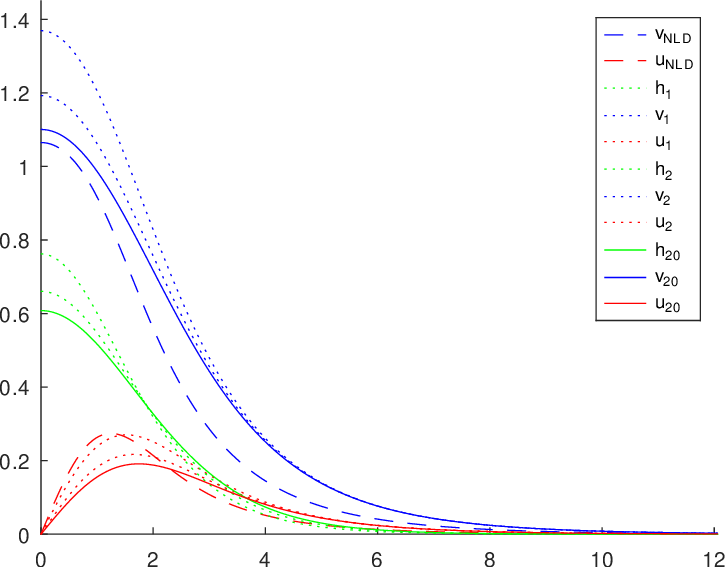}
\fi

\caption{\footnotesize
3D, the iterative method;
$m=1$, $\coupling=1$, $M=1$;
$\omega=0.5$ (left)
and
$\omega=0.9$ (right).
Solitary waves of cubic NLD
(dashed)
and iterations
of solitary waves of
DKG system:
first and second (dotted)
and
twentieth iterations are plotted.
}
\label{fig-3-1.00}
\end{figure}

\begin{figure}[!ht]
\ifpdf
\noindent\includegraphics[width=0.46\textwidth,height=140pt]{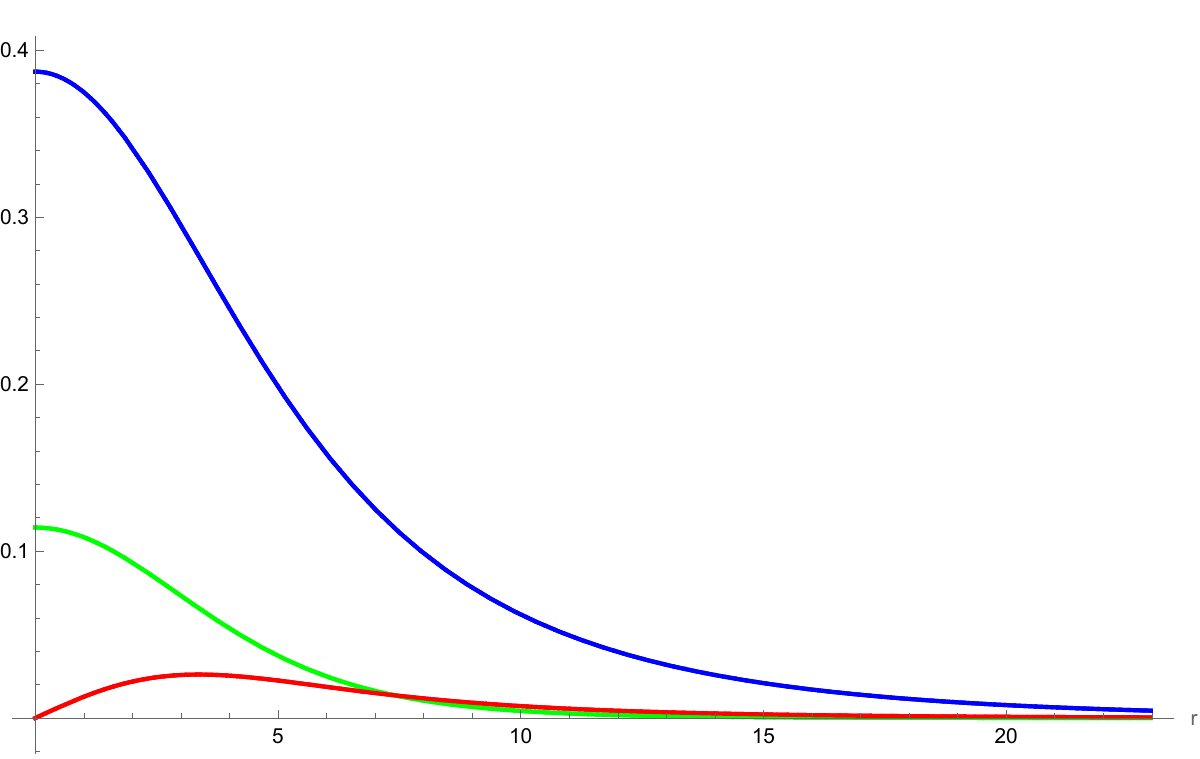}
\hfill\includegraphics[width=0.46\textwidth,height=140pt]{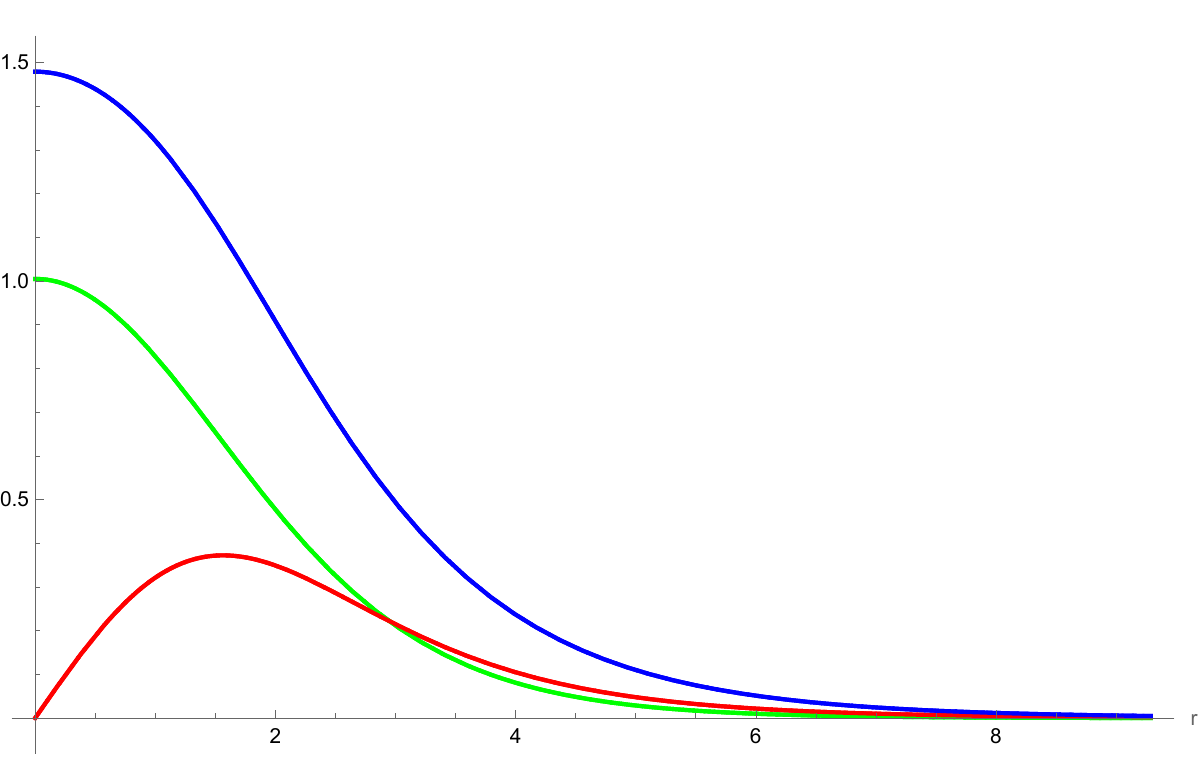}
\else
\noindent\includegraphics[width=0.46\textwidth,height=140pt]{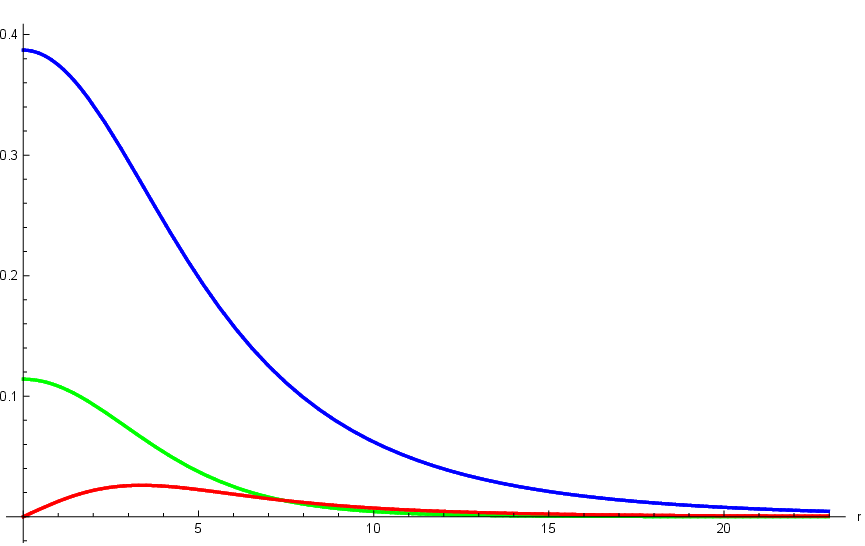}
\hfill\includegraphics[width=0.46\textwidth,height=140pt]{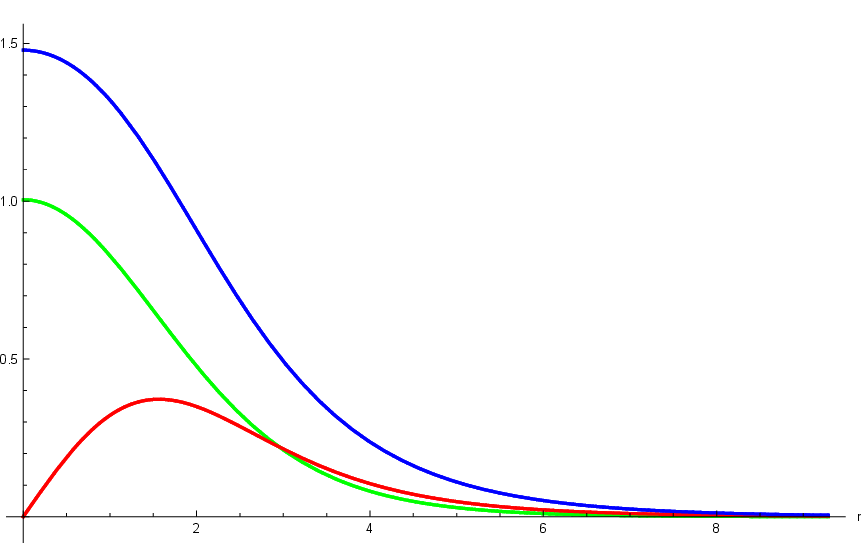}
\fi
    
\ifpdf
\noindent\includegraphics[width=0.46\textwidth,height=140pt]{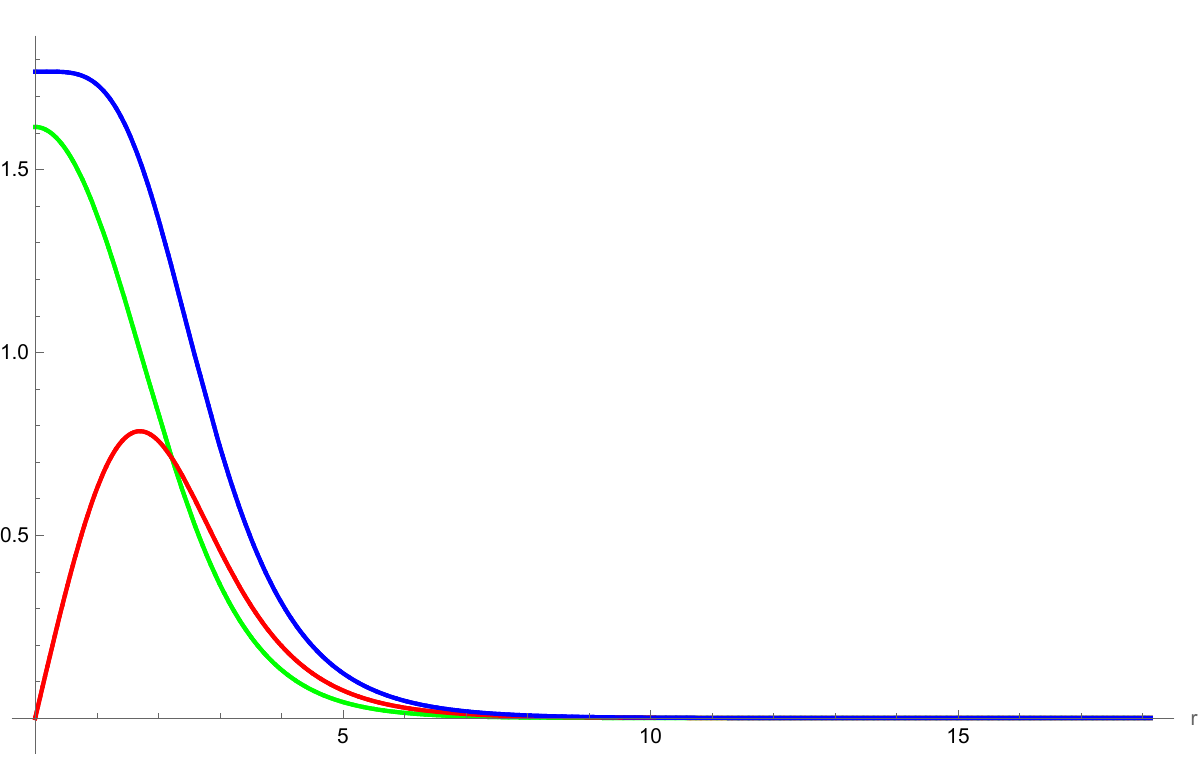}
\hfill\includegraphics[width=0.46\textwidth,height=140pt]{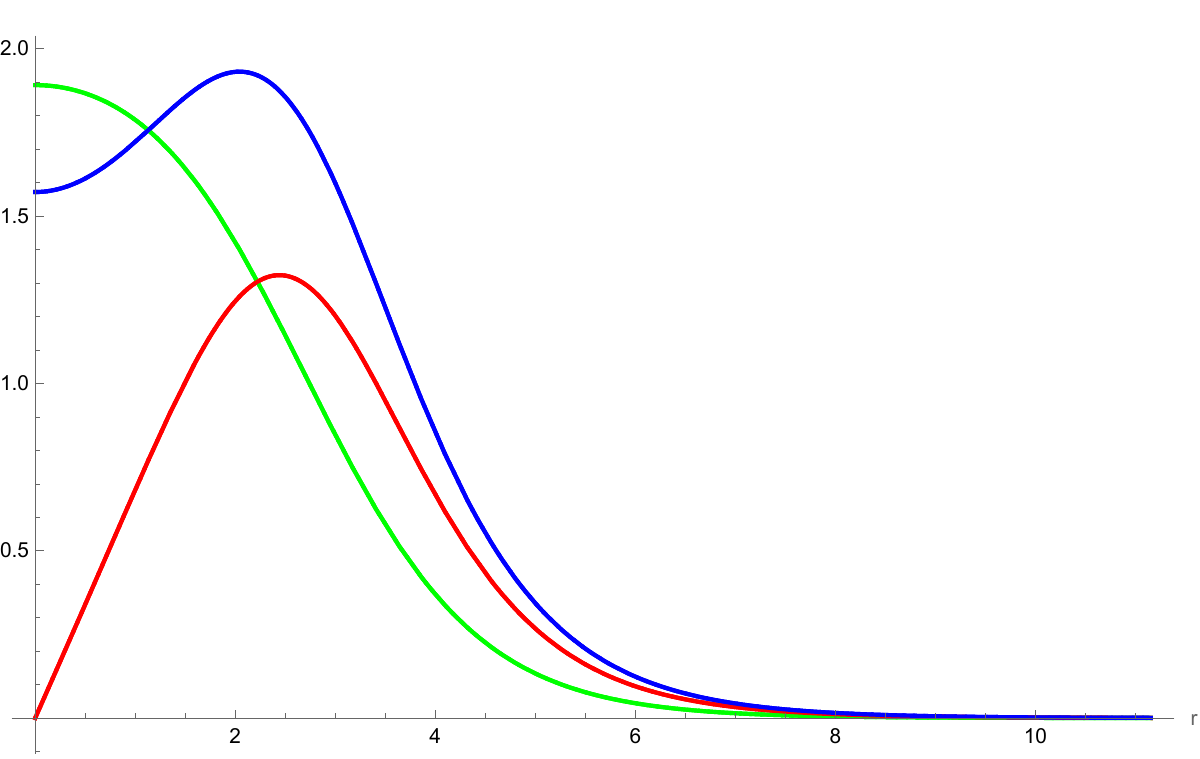}
\else
\noindent\includegraphics[width=0.46\textwidth,height=140pt]{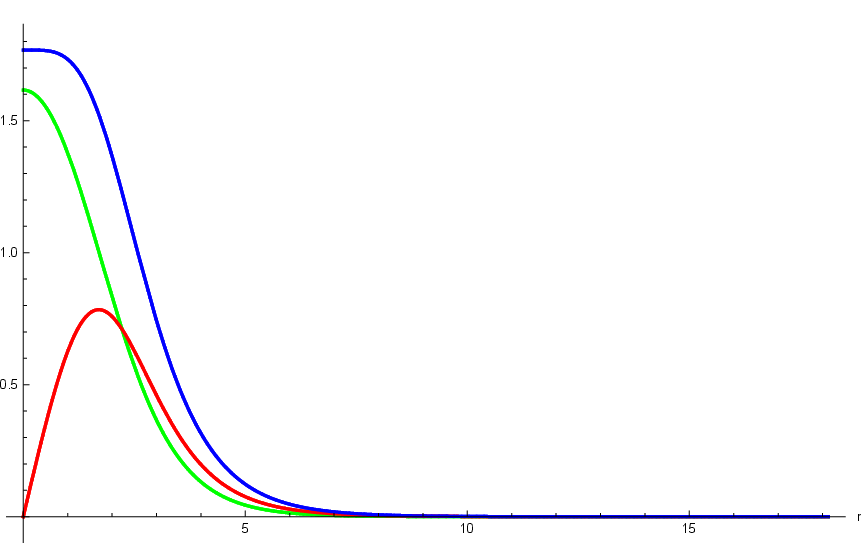}
\hfill\includegraphics[width=0.46\textwidth,height=140pt]{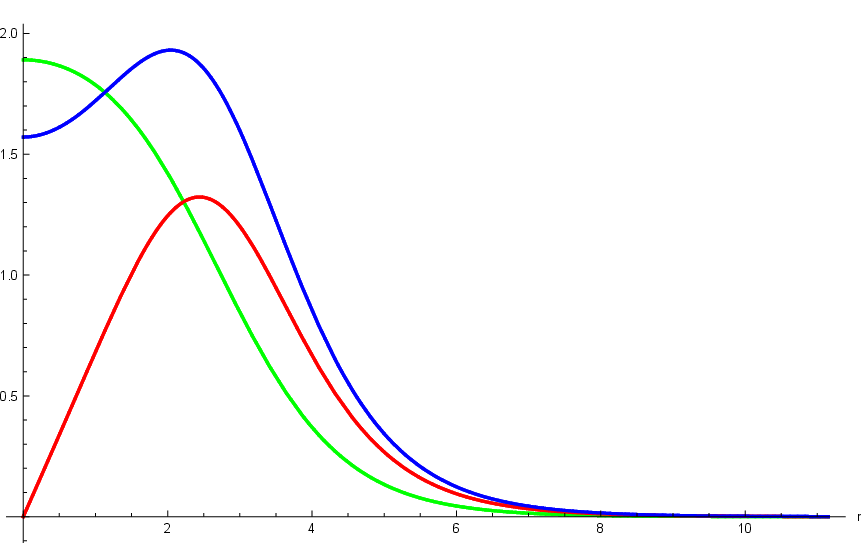}
\fi
 
\ifpdf
\noindent\includegraphics[width=0.46\textwidth,height=140pt]{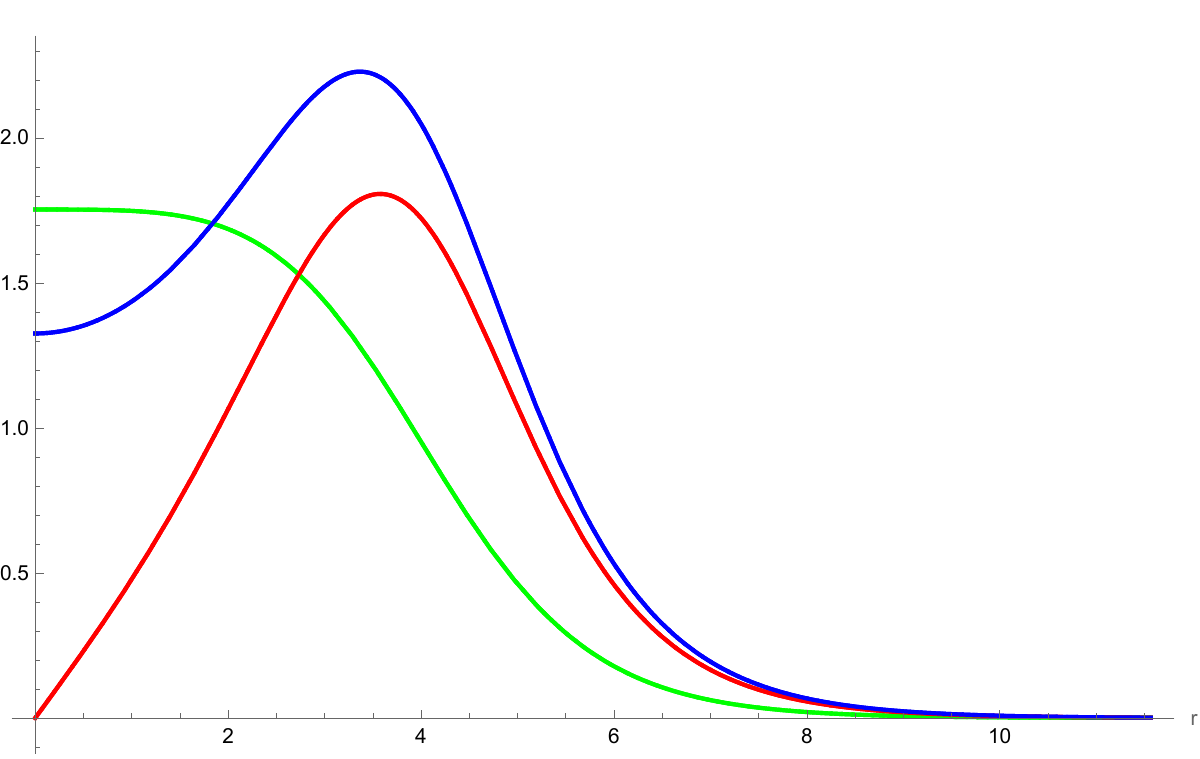}
\hfill\includegraphics[width=0.46\textwidth,height=140pt]{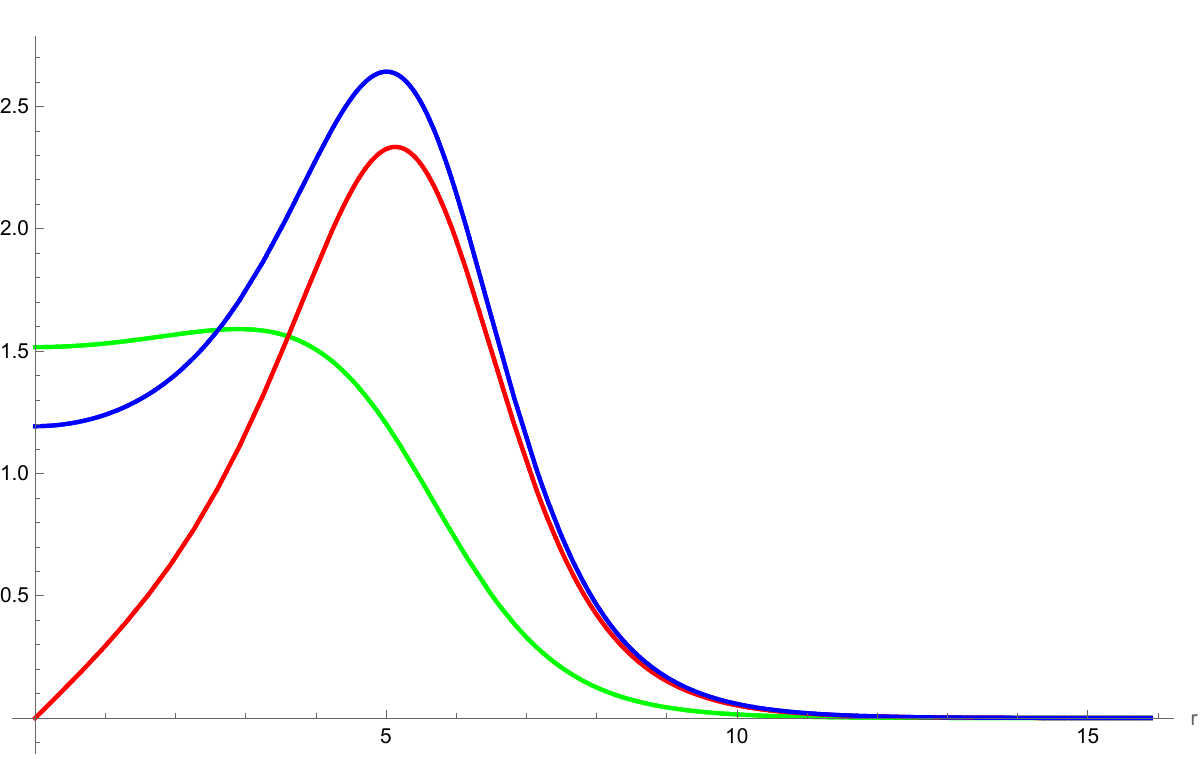}
\else
\noindent\includegraphics[width=0.46\textwidth,height=140pt]{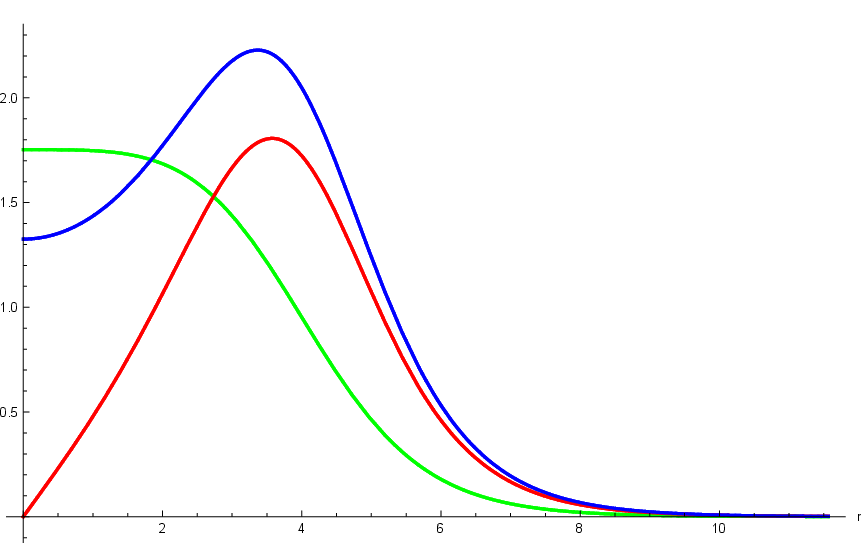}
\hfill\includegraphics[width=0.46\textwidth,height=140pt]{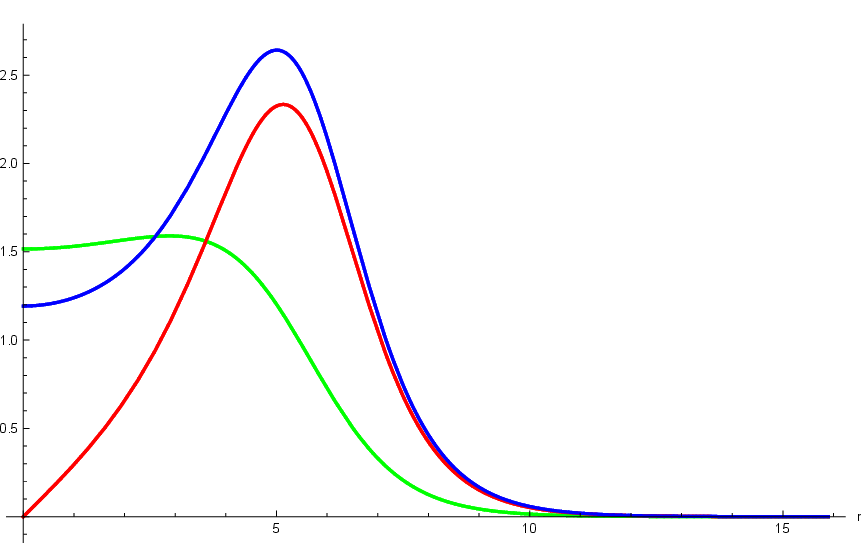}
\fi

\caption{\footnotesize
3D, the nested shooting method;
$m=1$,
$\coupling=1$, $M=1$;
the nested shooting method.
Profiles $v$ (blue), $u$ (red), and $h$ (green) of the solitary waves
corresponding to
$\omega=0.99$ (top left),
$\omega=0.8$ (top right),
$\omega=0.6$ (center left),
$\omega=0.4$ (center right),
$\omega=0.28$ (bottom left),
and
$\omega=0.20$ (bottom right). The functions $v$ and $h$ are monotonic decreasing respectively for $\omega\geq\omega_v\approx 0.6$ and for $\omega\geq\omega_h\approx 0.28$, and are otherwise first increasing and then decreasing.
    }
    \label{Figure-3D-NestedshootingMethod-M=1}
\end{figure}

The energy $E$ as a function of $\omega$
for several values of $M$
is presented on
Figure~\ref{fig-3d}
(the results from the iterative method and
the nested shooting method
are essentially indistinguishable on the plot).
For the corresponding data,
see Table~\ref{table-3d-iterations} as well as Tables~\ref{supp-table-3D-NestedshootingMethod-M=1}--\ref{supp-table-3D-SimpleShootingMethod-M=0} in the Supplementary Material.

The vertical dash-dot line
on the right plot
of Figure~\ref{fig-3d}
corresponds to $\omega_\star=0.936$,
which is
the location of the minimum
of $E$ (and $Q$) of solitary waves
of the cubic NLD.
As $M$ decreases, the minima of $E(\omega)$
move to the right,
approaching $\omega=m$
as $M\to 0$,
while the curve corresponding to $M=0$
goes to zero as $\omega\to m=1$.

\begin{remark}
\label{remark-like-nld}
One can see on Figure~\ref{fig-3d} (right)
that the plot of $E(\omega)$
corresponding to $M=1$ approaches
the NLD's $E(\omega)$
for $\omega$ near $m$.
One can also see this asymptotic behavior
in 1D
(see Figure~\ref{fig-1d} in
Appendix~\ref{sect_appendix_1D}).
This behavior is justified
by the scaling based on
the asymptotic behavior of
profiles of solitary waves
for $\omega\to m$,
when
$v(r)\sim \epsilon V(\epsilon r)$,
$u(r)\sim \epsilon^2 U(\epsilon r)$,
$h(r)\sim\epsilon^2 H(\epsilon r)$,
with $\epsilon=\sqrt{m^2-\omega^2}$,
hence the term $\Delta\phi$
in \eqref{dkg-stationary}
becomes negligible
in the limit $\omega\to m$
(cf. \cite[\S5]{comech2013polarons},
\cite{boussaid2017nonrelativistic}).
\end{remark}

\begin{figure}[!hbt]
\begin{center}
\ifpdf
\includegraphics[width=0.46\textwidth,height=140pt]{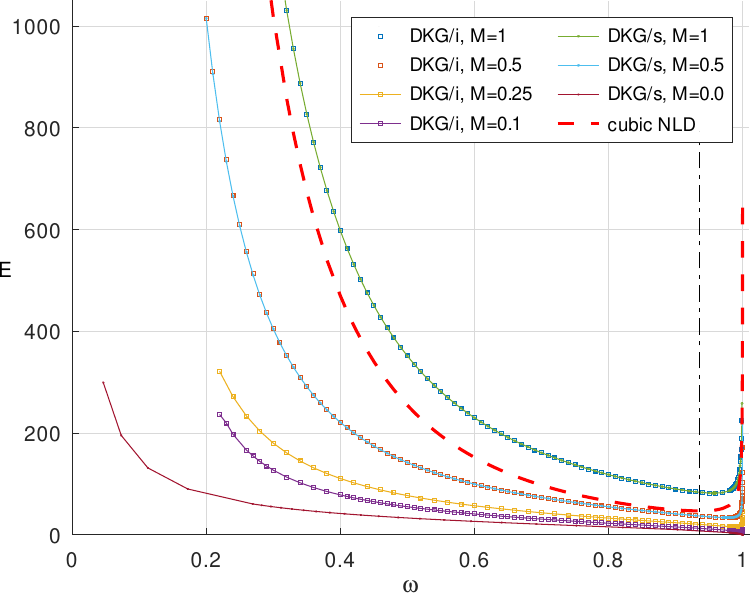}
\hfill\includegraphics[width=0.46\textwidth,height=140pt]{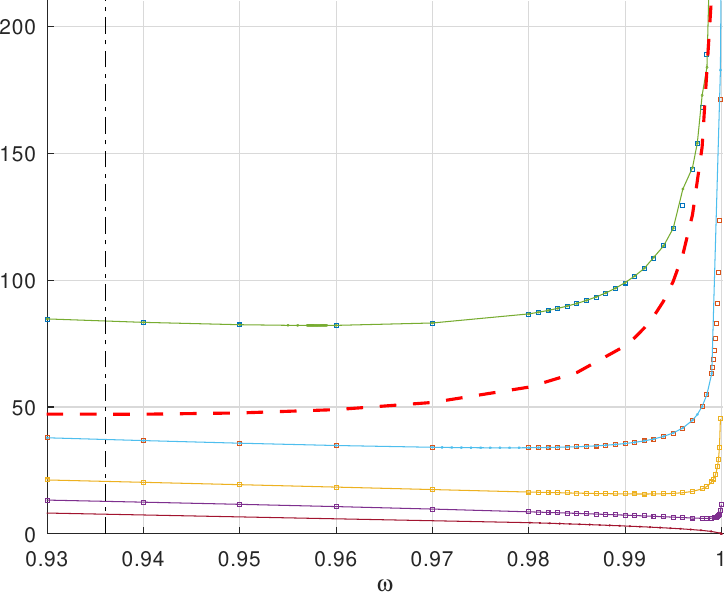}
\else
\includegraphics[width=0.46\textwidth,height=140pt]{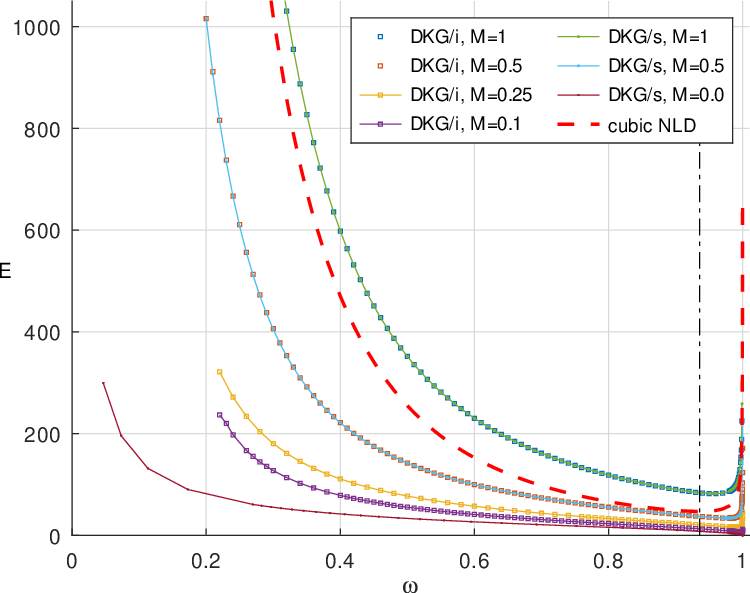}
\hfill\includegraphics[width=0.46\textwidth,height=140pt]{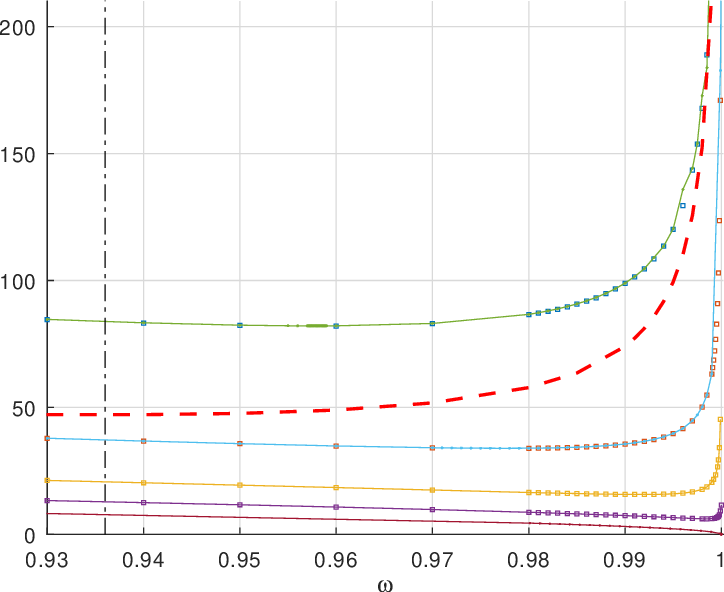}
\fi
\end{center}
\caption{\footnotesize
3D;
$m=1$, $\coupling=1$.
Left:
energy $E$
of solitary waves
of DKG system
as a function of $\omega$.
Right: magnified region
corresponding to $\omega\in(0.93,1)$.
Both iterative
and nested shooting methods for $M= 1$ and $M=0.5$;
the iterative method for $M=0.25$ and $M=0.1$;
the (standard) shooting method for $M=0$.
The squares (labeled ``DKG/i'') correspond to the iterative method and the dots (labeled ``DKG/s'') to the nested shooting method (the shooting method for $M=0$).
The dashed line corresponds to the solitary waves of the cubic NLD.
}
\label{fig-3d}
\end{figure}

\begin{remark}\label{remark-wakano-2}
We expect that both methods should allow to construct
solitary waves with nodes
(covered by the existence proof from \cite{esteban1996stationary}).
Since such solitary waves
are expected to be unstable,
we decided not to overload this
article with their construction.

On another hand,
for $M>0$,
we were not able to find
solitary waves
$\phi(x,\omega)e^{-\jj\omega t}$
with $\phi$
in the form of the second Wakano Ansatz~\cite{wakano-1966},
\begin{equation}\label{sw-2}
    \phi(x,\omega)
    =
    \begin{bmatrix}
        -\jj u(r,\omega)
        \begin{bmatrix}\cos\theta\\e^{\jj\upvarphi}\sin\theta\end{bmatrix}
        \\
        v(r,\omega)\begin{bmatrix}1\\0\end{bmatrix}
    \end{bmatrix},
\end{equation}
with $v,\,u$ satisfying the system
\begin{equation}\label{dkg-stationary-2}
    \begin{cases}
        \omega u=\p_r v+(m-\coupling h)u\,,
        \\
        \omega v=-\p_r u-\frac{n-1}{r}u
        -(m-\coupling h)v\,,
        \\
        \bigl(-\p_r^2-\frac{n-1}{r}\p_r+M^2\bigr)h
        =u^2-v^2\,,
    \end{cases}
    \qquad r>0\,.
\end{equation}
Since one would have $u$ vanishing at $r=0$,
the corresponding value of
$\phi^*\beta\phi=u^2-v^2$
would now be negative in a neighborhood of the origin.
\end{remark}

\subsubsection*{The iterative method}

We treat the 3D case
via one-dimensional FFT
(see Remark~\ref{remark-3d} below).
Here is our iterative
approach to finding solitary waves
in DKG system
numerically:
\begin{enumerate}
\item
Let $\omega\in(0,m)$.
Find the numerical solution
$\phi$
of the form \eqref{sw-1}
to the cubic NLD
(see \eqref{nld-stationary-1}
in Appendix~\ref{sect-nld},
where we take $f(\tau)=\tau$);
denote $v_0=v(0)$.
If we set
$h=v^2-u^2$ and $\coupling=1$,
we get the solution to
\begin{equation}\label{dkg-h}
    (D_0+(m-\coupling h)
    \beta-\omega)\phi=0\,.
\end{equation}

\item
Compute
$h=(-\Delta+M^2)^{-1}(v^2-u^2)$, via FFT, from
the solutions $v$ and $u$ just obtained.

\item
Using the computed field $h$,
adjust the coupling constant
$\coupling$
so that
$D_0+(m-\coupling h)\beta$
has an eigenvalue
$\omega$ corresponding to the ground state
of the form of the first Wakano Ansatz
\eqref{sw-1};
that is,
so that equation \eqref{dkg-h}
has a solution
of the form \eqref{sw-1}
with $v,\,u$ such that
\[
(v,u)\at{r=0}=(v_0,0)\,,
\qquad v\at{\R_{+}}>0\,,
\qquad \text{and} \qquad
\lim_{r\to+\infty}v(r)=0\,.
\]

\item
Repeat from Step 2 until
$\coupling$
and $h(0)$
stabilize,
and hence so do $\phi$ and $h$,
providing a solution to
\begin{equation}\label{alpha-general}
    \omega\phi
    =D_0\phi
    +(m-\coupling h)\beta\phi\,,
    \qquad
    (-\Delta+M^2)h 
    = \phi^*\beta\phi\,.
\end{equation}

\item
From
the computed solutions
$\phi$ and $h$
which satisfy \eqref{alpha-general}
get the solutions
to the system
\eqref{system-dkg-r}
with $\coupling=1$,
\begin{equation}\label{alpha-one}
    \omega\tilde\phi
    =D_0\tilde\phi
    +(m-\tilde h)\beta
    \tilde\phi\,,
    \qquad
    (-\Delta+M^2)\tilde h
    =\tilde\phi^*\beta\tilde\phi\,,
\end{equation}
by setting
\begin{equation}\label{alpha-change}
    \tilde\phi
    =\coupling^{1/2}\phi
    \qquad \text{ and } \qquad
    \tilde h=\coupling h\,.
\end{equation}
The charge, energy,
and other quantities
\eqref{def-kn}, \eqref{def-tvw}
corresponding to $\tilde\phi$ and
$\tilde h$
(which we denote
$\tilde Q$, $\tilde E$,
and so on)
are obtained from the charge, energy,
and other quantities
corresponding to $\phi,\,h$ by
\begin{equation}\label{alpha-q-e-change}
    \tilde Q=\coupling Q\,,
    \qquad
\tilde E=\coupling E\,,
\end{equation}
and also
$\tilde K=\coupling K$,
$\ \tilde N=\coupling N$,
$\ \tilde T=\coupling T$,
$\ \tilde V=\coupling V$,
\ and
$\ \tilde W=\coupling W$.
\end{enumerate}

\begin{remark}
\label{remark-3d}
The fast Fourier transform
works equally well
for finding DKG
solitary waves
in 1D and in 3D.
Let the density
$\sigma(x)=\phi(x)^*\beta\phi(x)$
be spherically symmetric
(so we consider it as a function of $r=\abs{x}$).
Then its Fourier transform
and the inverse Fourier
transform
are given by
\begin{align*}
\hat\sigma(\xi)
&=
\int_{\R^3}
e^{\jj\xi x}\sigma(x)\,dx
=2\pi\int_{\R_{+}}
\int_{-1}^1
e^{\jj\abs{\xi} r w}
\sigma(r)\,r^2\,dr\,dw
=
\frac{2\pi}{\jj\abs{\xi}}
\Big[
\int_{\R}
e^{\jj\abs{\xi} y}
\sigma(\abs{y})\,y\,dy
\Big],
\\
\sigma(x)
&=
\frac{2\pi}{(2\pi)^3}
\int_{\R_{+}}
\int_{-1}^1
e^{-\jj\lambda \abs{x} w}
\hat\sigma(\lambda)
\,\lambda^2\,d\lambda
\,dw
=\frac{\jj}{(2\pi)^2\abs{x}}
\int_{\R}
e^{-\jj k\abs{x}}
k\hat\sigma(\abs{k})
\,dk\,.
\end{align*}
Above,
$w=\cos\theta$,
with $\theta$ the angle between
$x$ (or $y$) and $\xi$;
for $\lambda>0$,
$\hat\sigma(\lambda)$
denotes
the value of $\hat\sigma$
on any $\xi\in\R^3$ with $\abs{\xi}=\lambda$.
Thus, both the direct and inverse
Fourier transform reduce to the
one dimensional case
and the fast Fourier transform
can be applied
in 3D
with no impact on the
computation speed.
\end{remark}

The solutions were constructed in
Octave with the {\tt ode45} ODE solver.
We used the length
$L=10/\sqrt{m^2-\omega^2}$,
to take into account weaker decay of solitary waves
with $\omega$ near $m=1$.

\begin{remark}
Since $h\sim e^{-M\abs{x}}/\abs{x}$,
this value of $L$ is
insufficient for small
values of $M$
and would contribute
to larger
relative errors when evaluating
$T(h)$ and $W(h)$
(see \eqref{def-tvw}).
\end{remark}

The step was set at $\varDelta x=0.01$;
absolute and relative tolerances were set to $10^{-8}$.

The energy
as a function of $\omega$
are represented
on Figure~\ref{fig-3d}
(three spatial dimensions)
and on Figure~\ref{fig-1d}
(one spatial dimension)
for several values of $M$.  

The solitary wave profiles for two values of $\omega$
are plotted on
Figure~\ref{fig-1-1.00}
(1D, $m=M=1$)
and
Figure~\ref{fig-3-1.00}
(3D, $m=M=1$).
The corresponding initial data and values of energy and charge
(recomputed for $\coupling=1$)
are represented in Tables~\ref{table-1d-iterations}
and~\ref{table-3d-iterations}.
The energy
is computed by
$E=\omega Q(\phi)-V(\phi,h)/2$
(see \eqref{dkg-energy-no-derivatives}).
For the estimate of the error,
we use the quantity
\begin{align}
\label{def-varepsilon}
\varepsilon=
\omega Q-\frac{n-1}{n}K-N-\frac{n+2}{2n}V
-\frac{2}{n}W\,,
\end{align}
which is supposed to be zero
in view of the virial identity
\eqref{for-epsilon}.
Step zero
corresponds to the solitary wave
of the nonlinear Dirac equation
which we use for the first iteration; the error there
corresponds to the error
\eqref{def-varepsilon-nld}
in the virial identity for the
nonlinear Dirac equation
(see Appendix~\ref{sect-nld}).
We note that
for the iterative method
in the 3D case with $M=1$, $\omega=0.5$
(see Table~\ref{table-3d-iterations}),
the relative error in the virial identity,
$\varepsilon/(\omega Q)$, is approximately $0.05\%$;
for $\omega=0.9$,
$\varepsilon/(\omega Q)$ is approximately $0.01\%$.

\subsubsection*{The nested shooting method}

We are to solve~\eqref{system-dkg-r} with initial values $v(0)=v_0$,
$u(0)=0$,
$h(0)=h_0$, $h'(0)=0$
and therefore need to determine $v_0$ and $h_0$ such that the solution $(v, u, h)$ for these two initial values satisfies $v, u, h \to 0$ at infinity.
We will do this adjusting
$v_0$ and $h_0$ via the bisection method
(also known as the dichotomy method or the interval halving method).
While we cannot
adjust two parameters by the shooting method at the same time,
we nest one shooting inside the other: the outer shooting is performed on $h_0$ and for the each tested value of $h_0$ the inner shooting adjusts $v_0$.
Even though we do not see
an \emph{a priori}
reason for such an algorithm to systematically converge, we have a set of rules
for which the convergence takes place.
We make the decision on adjusting
$v_0$ based on certain criteria on $h$ and, conversely, the adjustment of $h_0$
is based on the criteria on $v$.

Here is our approach.
We fix the parameters $(m,M,\coupling,\omega)$.
We perform the following (outer)
shooting to find $h_0$
using the bisection method,
starting with two values $h_0^-$ and $h_0^+$,
which we set anywhere to $(m-\omega)/\coupling$ and $2$, respectively
(see Remark~\ref{remark-small-omega}
below).

\begin{enumerate}[label=\arabic*)]
    \item\label{DoubleShootingOuterDichoItem1}
Set $h_0=(h_0^- + h_0^+)/2$ and for this fixed value of $h_0$ perform the following (inner) shooting on $v_0$
between two values $v_0^-$ and $v_0^+$,
which we initially set at
$0$ and $2$, respectively
(see also Remark~\ref{remark-small-omega}
below):
\begin{enumerate}[label=\roman*)]
\item\label{DoubleShootingInnerDichoItem1}
Set $v_0=(v_0^- + v_0^+)/2$
and
solve~\eqref{system-dkg-r} numerically with initial values $(v_0, h_0)$
in order to obtain a solution $(v, u, h)$ to the system on the interval $[0,L]$. We used the Mathematica's solver \texttt{NDSolve} for this.
Here, $L$ is the smallest range for which both $s_0$ and $s_1$ (see next step) are not $+\infty$ or,
when \texttt{NDSolve} fails to reach a large enough range so that $s_0, s_1<+\infty$, it is the maximum range on which the solver produces the solution. In particular, and contrarily to the iterative method, $L$ is not a parameter fixed \emph{a priori}.
            
\item\label{DoubleShootingInnerDichoItem2}
Set
\begin{align*}
\hskip -10mm
s_0:= \inf\bigl\{r \in (0,L) : h(r)=0\bigr\}
\quad
\mbox{and}
\quad
s_1:= \inf\bigl\{r \in (0,L) : h'(r)=0\,, \ h''(r)\geq0\bigr\};
\end{align*}
if $s_0<s_1$, set $v_0^+=v_0$;
otherwise, set $v_0^-=v_0$ (see Remark~\ref{remark-dicho-conditions}
below).
\item\label{DoubleShootingInnerDichoItem3}
Go back to Step~\ref{DoubleShootingOuterDichoItem1}-\ref{DoubleShootingInnerDichoItem1}
and repeat until  $(v_0^--v_0^+)/v_0$ reaches the desired relative precision.
\end{enumerate}
\item\label{DoubleShootingOuterDichoItem2}
Set
\begin{align*}
    t_0 := \inf\bigl\{r \in (0,L) : v(r)=0\bigr\}
\quad
\mbox{and}
\quad
    t_1 := \inf\bigl\{r \in (0,L) : v'(r)=0\,, \ v(r) < v(0)\bigr\}\,;
\end{align*}
if $t_0<t_1$, set $h_0^+=h_0$;
otherwise, set $h_0^-=v_0$ (see Remark~\ref{remark-dicho-conditions}
below).
\item\label{DoubleShootingOuterDichoItem3} Go back to Step~\ref{DoubleShootingOuterDichoItem1} and repeat until $(h_0^--h_0^+)/h_0$ reaches the desired relative precision.
\end{enumerate}

For each $\omega$,
we compute the charge $Q$
defined by~\eqref{def-Q};
$E$ is computed via~\eqref{dkg-virial-noderivatives-with-E}.
In order to control the accuracy of our numerics, we compute the relative error
in the virial identity;
we rewrite \eqref{def-varepsilon}
in the form
that does not contain
$K(\phi)$ -- cf. \eqref{dkg-virial-noderivatives} -- as
$\varepsilon = \omega Q - N - \frac{4-n}{2} V - 2 W$,
and now the relative error in 3D reads as follows:
\begin{equation}\label{relative_error_shooting_3D}
\frac{|\varepsilon|}{\omega Q} =
\left| 1 -  \frac{N + V/2 + 2 W}{\omega Q} \right|.
\end{equation}
For the relative precision on $v_0$ and $h_0$
tested respectively in Steps~\ref{DoubleShootingOuterDichoItem1}-\ref{DoubleShootingInnerDichoItem3}
and~\ref{DoubleShootingOuterDichoItem3},
we take
the threshold of $10^{-10}$
for $(m, M, \coupling) = (1,1,1)$ and $10^{-15}$ for $(m, M, \coupling) = (1,1/2,1)$ and $(m, M, \coupling) = (1,1/4,1)$;
we take a smaller value only when we do not reach a small enough relative error in the virial identity with our default choice. This occurs for instance for the values of $\omega$
near $m$
and also for (some of) the values of
$\omega$ that are close to
$\omega_v\approx 0.6$
and $\omega_h\approx 0.28$
for which $v''(0)$ and $h''(0)$,
respectively, are equal to zero
(see Fig.~\ref{Figure-3D-NestedshootingMethod-M=1}, center left and bottom left,
respectively).
For example,
for $(m, M, \coupling) = (1,1,1)$, we had
to increase precision at $\omega = 0.6$ and $\omega=0.61$
and at $\omega=0.26$.

\begin{remark}\label{Rmk_nested_shooting_heuristic}
    The heuristic argument behind the decision ``if $h$ vanishes, then $v_0$ needs to be lowered'' is the following.
    From the third equation in~\eqref{system-dkg-r}, we have
    \[
        n h''(0) = - v_0^2 + M^2 h_0\,.
    \]
    Thus, for a fixed $h_0$, by making smaller $v_0>0$, $h''(0)$ is made larger. The heuristic is then, still for a fixed $h_0$, that if $h$ vanishes, then making $h''(0)$ larger will have the new $h$ vanishing later or not at all.
\end{remark}

\begin{remark}
\label{remark-small-omega}
For our application of this method, we made the
\emph{a priori} assumption on the solitary wave that $u'(0)>0$;
by the first equation
from~\eqref{system-dkg-r},
this is equivalent to
$h(0) > \frac{m-\omega}{\coupling}$. Nevertheless, $h_0^-$ can be taken anywhere in $\left(0, (m-\omega)/\coupling \right]$ without impact on the convergence of the algorithm to the solitary wave.

For smaller values of $\omega$
(e.g., for $\omega\le 0.24$
in the case $m=M=\coupling=1$)
we have to fine-tune our choice of initial values $v_0^\pm$ and $h_0^\pm$ -- based on $(v_\omega(0), h_\omega(0))$ found for larger values of $\omega$
-- in order for the method to converge to a solitary wave.
\end{remark}

\begin{remark}
\label{remark-dicho-conditions}
The set of rules is basically ``if $h$ vanishes, then $v_0$ needs to be lowered'' and ``if $v$ vanishes, then $h_0$ needs to be lowered'', with the implemented conditions involving the zeros of $h'$ and $v'$ being technicalities to address the fact that the numerical solutions are of course not computed on the whole half-line.

The condition on $h$ in Step~\ref{DoubleShootingOuterDichoItem1}-\ref{DoubleShootingInnerDichoItem2} is to set $v_0^+=v_0$ if $h$ vanishes (strictly) before $h'$ on $(0,+\infty)$, and $v_0^-=v_0$ otherwise, with the subtlety that we exclude the zeros of $h'$ at which $h''$ is negative. This is necessary for finding solitary waves for which $h$ is increasing on some interval $(0,R)$ and then decreasing on $(R, +\infty)$, which is the case for smaller
values of $\omega$ (see Fig.~\ref{Figure-3D-NestedshootingMethod-M=1}).

There is a similar situation for the condition on $v$ in Step~\ref{DoubleShootingOuterDichoItem2}, allowing us to consider
functions $v$ which are
increasing on $(0,R)$ for some $R\geq0$
and decreasing afterwards.
\end{remark}

\begin{remark}
There is an alternative nested
shooting approach
 that also converges.
 It is similar to the one
 we described above
but with an outer shooting on $v_0$ where the criterion to redefine $v_0^\pm$ is based on zeros of $v$ and $v'$, and with an inner shooting on $h_0$ where the criterion to redefine $h_0^\pm$ is
based on the zeros of $h$ and $h'$. While this approach might seem more natural because the criteria to shoot on the initial value of a function is based on the zeros of the function itself and its derivative, it does not yield better results in our context. Indeed, when applying it on decreasing the values of $\omega$, the first value for which it stops to converge (without fine-tuning) is larger than for the approach described previously. Moreover, for the values of $\omega$ for which it converges, it gives similar results in terms of accuracy and speed. We therefore decided not to overload this
    article with this alternative approach.
\end{remark}

\section{The case of massless scalar field in 3D}
\label{sect-dkg-numerics-massless}

In the case of massless scalar field,
we can also compute the solitary waves
with high accuracy
via the (standard) shooting method.
Unlike in the $M>0$ case
(where $h$ decays to zero exponentially),
if
$(\phi,h)$
is a solution to \eqref{system-dkg-r} with $M=0$,
then so is
$(\phi,h+c)$ for any $c\in\R$;
now \eqref{dkg-stationary} takes the form
\begin{align}\label{dkg-stationary-0}
\begin{cases}
\omega\phi=D_0+(m-\coupling h)\beta\phi,
\\
-\Delta h=\phi^*\beta\phi,
\end{cases}
\end{align}
and one needs to set some condition
to select the solution $h$
to the above system.
Given that the solution $h$
to \eqref{dkg-stationary} with localized
$\phi$ satisfies $\lim_{r\to+\infty}h(r)=0$,
we will use the same condition
for \eqref{dkg-stationary-0}.
At the same time, we can introduce a new
function
\begin{align}\label{def-H}
H(r) := h(r) - \frac{m}{\coupling},
\end{align}
and solve the system
\begin{align}\label{dkg-stationary-s}
\begin{cases}
\omega\phi=D_0-\coupling H\beta\phi,
\\
-\Delta H=\phi^*\beta\phi,
\end{cases}
\end{align}
specifying $H(0)=s$,
and then determine the ``effective'' mass
by
(cf. \eqref{def-H})
\begin{align}\label{def-ms}
m:=-\coupling\lim_{r\to+\infty}H(r)
\end{align}
and restore from \eqref{def-H}
the scalar field
\begin{align}\label{def-hs}
h(r)=H(r)+\frac{m}{\coupling},
\end{align}
so that indeed
$h(r)\to 0$ as $r\to+\infty$.

To summarize, in the case
$M=0$ we can reduce
the shooting of the solution to~\eqref{dkg-stationary-0} with respect to two parameters,
$v(0)$ and $h(0)$, to the shooting of the solution to~\eqref{dkg-stationary-s} with respect to only one parameter,
$v(0)$,
after fixing some $H(0)=s$.
We then
obtain the corresponding
``effective value'' of $m$
and the scalar field $h(r)$
\emph{a posteriori}
from~\eqref{def-ms}
and~\eqref{def-hs}.
The solution can then be scaled
from the ``effective value'' of mass
to any desired value.

\subsubsection*{The (standard) shooting method}

Thus, the procedure for the shooting method
is as follows.
We consider the system
\eqref{dkg-stationary-s},
which we write in the form
(cf. \eqref{system-dkg-r})
\begin{equation}\label{system-dkg-r-2}
\begin{cases}
  \p_r u
  =
  -\frac{n-1}{r}u+(\omega_0+\coupling H)v\,,
  \\
  \p_r v=
  -(\omega_0-\coupling H)u\,,
  \\
  -\left(
  \p_r^2+\frac{n-1}{r}\p_r
  \right)H=v^2-u^2\,,
\end{cases}
\qquad
r>0\,,
\end{equation}
with some fixed $\omega_0>0$;
for definiteness, we set
\begin{align}\label{omega-is-1}
\omega_0=1\,.
\end{align}
We fix some value of the parameter $s\in\R$
and solve the system
\eqref{system-dkg-r-2}
with the initial data
\[
u(0)=0,
\qquad
H(0)=s,
\qquad
H'(0)=0,
\]
using the shooting method
to adjust the value of $v(0)>0$
so that $v(r)\to 0$
as $r\to+\infty$
(this practically
implies that $u(r)\to 0$, too),
and denoting the solution that
we found by
$\big(u_s(r),\,v_s(r),\,H_s(r)\big)$,
$r\ge 0$.
We find such
decaying solutions for
$s>-\omega_0/\coupling=-1$;
cf. Table~\ref{supp-table-3D-SimpleShootingMethod-M=0}.)
The corresponding value $m_s$
and the scalar field $h_s$
corresponding to this particular value of
$s\in\R$
are restored via \eqref{def-ms}
and \eqref{def-hs}:
\[
m_s
:=-\coupling\lim_{r\to+\infty}H_s(r),
\qquad
h_s(r)
:=H_s(r)+\frac{m_s}{\coupling}.
\]

\begin{remark}
\label{remark-41}
For the field $h$ and
the value
of the effective mass
defined in \eqref{def-hs}
and \eqref{def-ms},
one needs to have the value of 
$H_s(+\infty)$.
Since the decay of
\[
h_s(x,\omega)=(4\pi\abs{x})^{-1}\ast(v_s^2-u_s^2)
\sim \abs{x}^{-1}
\]
is relatively slow
(we assume that $v_s$ and $u_s$
are exponentially localized),
when our shooting method returns us
the function on some interval
$[0,L]$,
we use the approximation
\begin{align*}
H_s(+\infty)
&=\lim\sb{r\to+\infty}H_s(r)
=\lim\sb{r\to+\infty}
(H_s(L)-h_s(L)+h_s(r))
\\
&=H_s(L)-h_s(L)
    \approx
    H(L)-\frac{1}{4\pi L}
    \int_{\abs{x}\le L}
    \phi_s^*\beta\phi_s\,,
\end{align*}
where $L$
is large enough so that
the integrand
in the last term
is essentially supported
in the ball
of radius $L$
in $\R^3$, and where $\phi_s:=(v_s, u_s)$.
\end{remark}

\begin{remark}
In all the numerics we have only
seen the values $m_s>0$;
for the simplicity of exposition, 
from now on we assume that
$m_s>0$.
\end{remark}

Since
$\phi_s=(v_s, u_s)$ and $h_s$
satisfy
\[
\omega_0\phi_s
=D_0\phi_s-\coupling H_s\beta\phi_s
=D_0\phi_s+(m_s-\coupling h_s)\beta\phi_s\,,
\quad
-\Delta H_s=-\Delta h_s=\phi_s^*\beta\phi_s\,,
\quad
\lim_{r\to+\infty}h_s(r)=0\,,
\]
the couple
$(\hat\phi(x),\ \hat h(x))
=\bigl(k\phi_s(k x),\ h_s(k x)\bigr)$,
with $k>0$,
satisfies
\[
k\omega_0 \hat\phi
= D_0\hat\phi
+(km_s - k\coupling \hat{h}) \beta\hat\phi\,,
\qquad
-\Delta \hat{h}=\hat\phi^*\beta\hat\phi,
\qquad
\lim_{r\to+\infty}\hat h(r)=0\,,
\]
and, by virtue of~\eqref{alpha-general}--\eqref{alpha-change},
the couple
$(\phi(x),\  h(x)\bigr)
= \left(k^{3/2}\phi_s(k x),\ k h_s(k x)\right)$
satisfies
\[
k\omega_0\phi = D_0\phi + \left( k m_s - \coupling h \right) \beta\phi\,,
\qquad
-\Delta h = \phi^*\beta\phi\,,
\qquad
\lim_{r\to+\infty}h(r)=0\,.
\]
For a given $m>0$,
we set $k = m/m_s>0$, so that the triple
\begin{equation}\label{rg}
\bigl(\phi(x),\ h(x),\ \omega\bigr)
=
\left(
\left(\frac{m}{m_s}\right)^{\frac{3}{2}} \phi_s\left(\frac{m}{m_s} x \right),
\ \frac{m}{m_s} h_s\left(\frac{m}{m_s} x \right),
\ \frac{m}{m_s}\omega_0
\right)
\end{equation}
solves
\[
\omega \phi = D_0\phi
+ (m-\coupling h)\beta\phi\,,
\qquad
-\Delta h = \phi^*\beta\phi\,,
\qquad
\lim_{r\to+\infty}h(r)=0\,.
\]
The quantities
\eqref{def-Q},
\eqref{def-kn}, and
\eqref{def-tvw}
for the triples
$(\phi_s,h_s,m_s)$
and
$(\phi,h,m)$
are related by
\[
Q = Q_s\,,
\qquad
K = \frac{m}{m_s} K_s\,,
\qquad
N = \frac{m}{m_s} N_s\,,
\qquad
V = \frac{m}{m_s} V_s\,,
\quad\text{ and } \quad
T = \frac{m}{m_s} T_s
\]
(while $W_s = 0 = W$ due to $M=0$),
hence
by~\eqref{dkg-virial-noderivatives-with-E}
the corresponding
energies and the errors in the
virial identity
are related by
\[
    E = \frac{m}{m_s} E_s
    \quad \text{ and } \quad
    \varepsilon
    =
    \omega Q
    - N - \frac{1}{2} V
    =\frac{m}{m_s}\varepsilon_s\,.
\]
We apply the transformation \eqref{rg}
to rescale all our solitary waves
to the same value $m$,
which finally gives us the corresponding value of $\omega$:
\begin{equation}\label{def-new-omega}
\omega=\frac{m}{m_s}\omega_0\,.
\end{equation}

The above method allows us to construct
solitary waves of the form \eqref{sw-1}
to the system~\eqref{dkg-stationary} with 
$m=1$,
$\coupling=1$, and $M=0$ for the whole range $\omega \in (0,m)$.
(Let us mention that in \eqref{omega-is-1}
we were free to choose $\omega_0=1$,
and then \eqref{def-new-omega}
takes the form
$\omega=1/m_s$.)
The data obtained by the shooting method on $v(0)$
for different values of $H(0)=s$
is collected in
Table~\ref{supp-table-3D-SimpleShootingMethod-M=0}
in the Supplementary Material.
For the estimate of the accuracy,
we use the relative error
in the virial identity
in the form~\eqref{relative_error_shooting_3D},
which in the case
$n=3$
and
$M=0$ -- hence $W(h)=0$, cf. \eqref{def-tvw} --
reduces to
\[
\frac{\abs{\varepsilon}}{\omega Q}
=
\bigg\vert
1 -  \frac{N + V/2}{\omega Q}
\bigg\vert.
\]
The corresponding values of the energy
are superimposed on Figure~\ref{fig-3d}
over the values of $E(\omega)$
corresponding to massive scalar field
(for several different positive values of $M$).
The profiles of the solitary waves 
for several values of $\omega$
in the case of the massless spinor
field
are displayed in Figure~\ref{Figure-profiles-3D-SimpleShootingMethod-M=0}.

\begin{remark}
Above, to construct different solitary
waves,
we were fixing $\omega_0=1$
and, for several initial values $H(0)=s$ for the scalar field,
shooting with respect to $v(0)$.
Let us mention that
fixing instead the initial value of the scalar field to $H(0)=0$
and, for several values of $\omega_0$,
shooting with respect to $v(0)$
leads to one and the same
value of $\omega$
in \eqref{def-new-omega},
and thus is not suited for constructing
solitary waves for different $\omega\in(0,m)$.
Indeed,
if the triple
$(\phi(x),\,H(x),\,\omega_0)$
satisfies the system
\[
\omega_0\phi
=
D_0\phi-\coupling H\beta\phi,
\qquad
-\Delta H=\phi^*\beta\phi,
\]
then so does the triple
\[
(\tilde\phi(x),~\tilde H(x),~\tilde\omega)=
\bigl(
k^{3/2}\phi(kx),~k H(kx),~k\omega_0
\bigr),
\qquad
k>0\,;
\]
since the effective masses,
$m=-\coupling\lim_{r\to+\infty}H(r)$
and
$\tilde{m}
=-\coupling\lim_{r\to+\infty}\tilde H(r)$,
are related by
$\tilde{m}=km$
(since so are $\tilde H$ and $H$),
one concludes that
$\omega = \omega_0/m$
in~\eqref{def-new-omega}
is the same as
$\omega=\tilde\omega/\tilde{m}$.
\end{remark}

\begin{figure}[!ht]
\ifpdf
\noindent\includegraphics[width=0.46\textwidth,height=140pt]{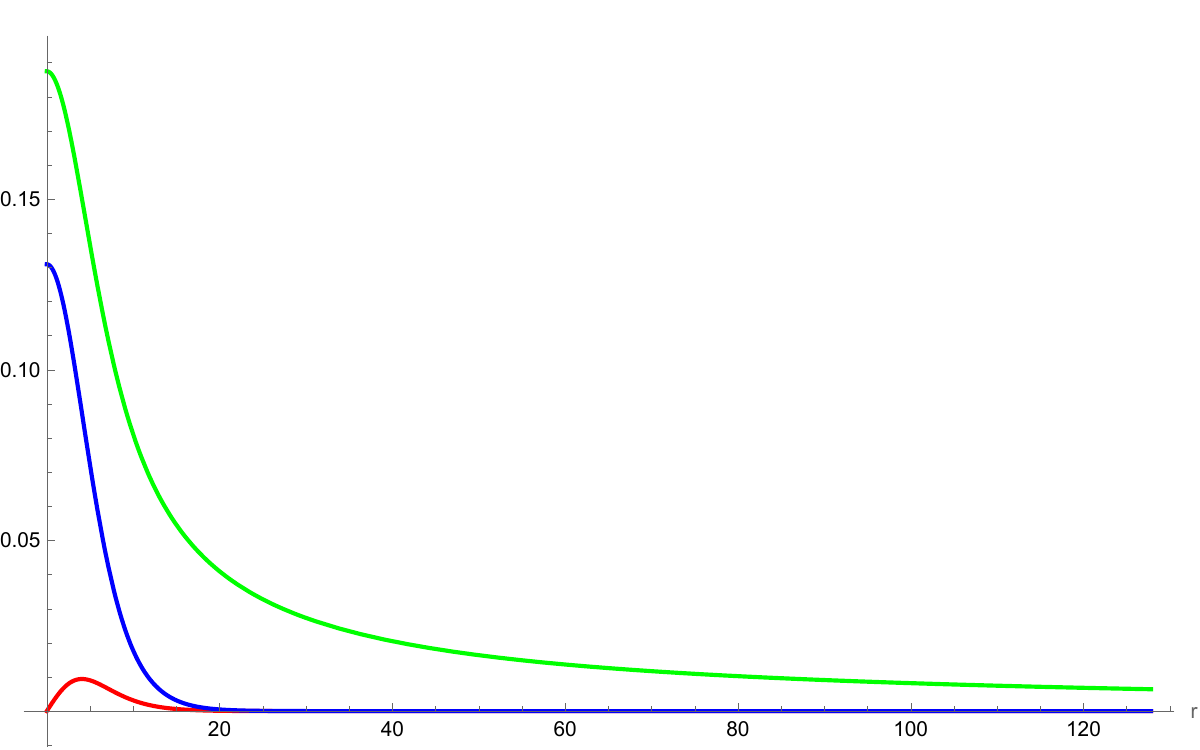}
\hfill\includegraphics[width=0.46\textwidth,height=140pt]{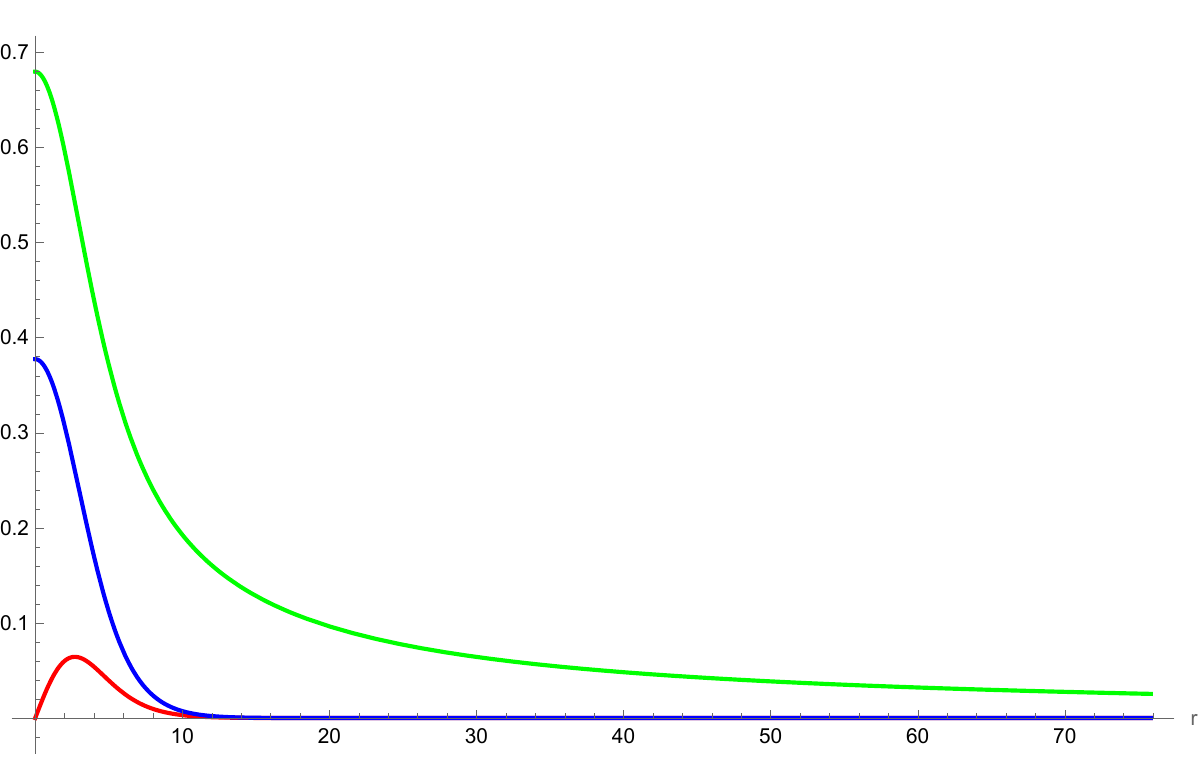}
\else
\noindent\includegraphics[width=0.46\textwidth,height=140pt]{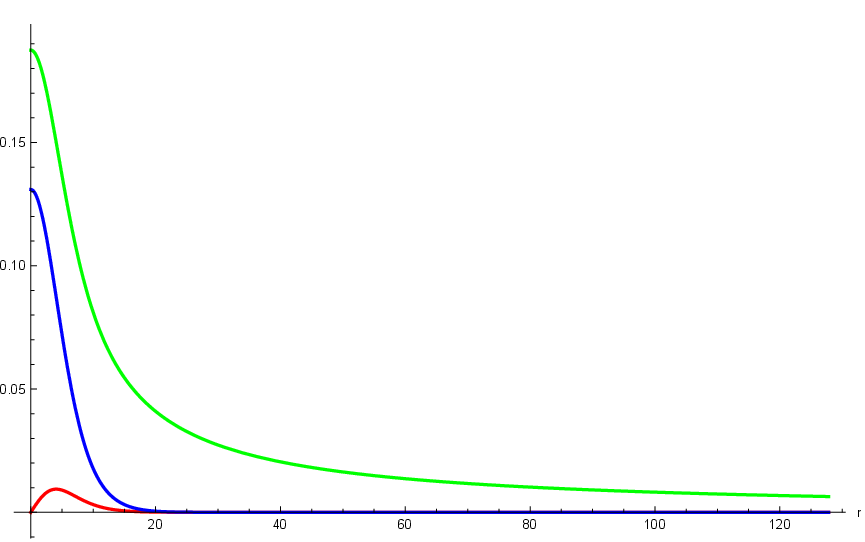}
\hfill\includegraphics[width=0.46\textwidth,height=140pt]{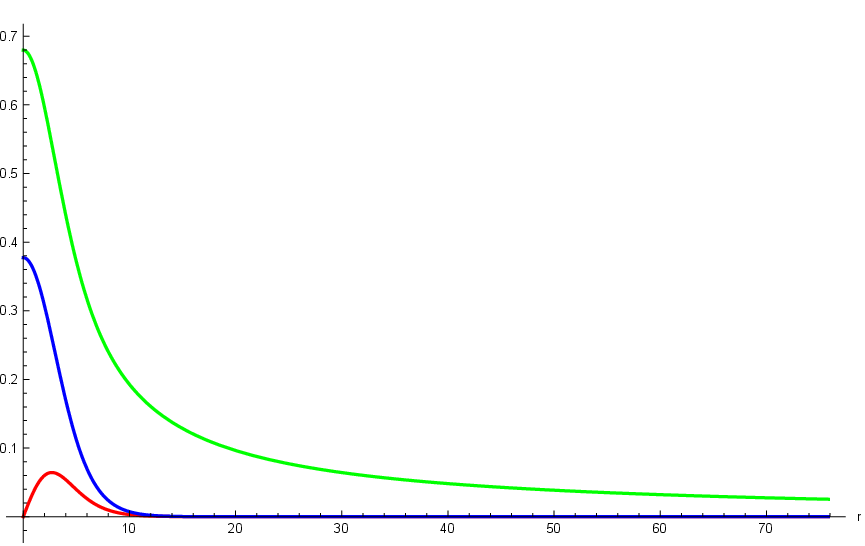}
\fi
    
\ifpdf
\noindent\includegraphics[width=0.46\textwidth,height=140pt]{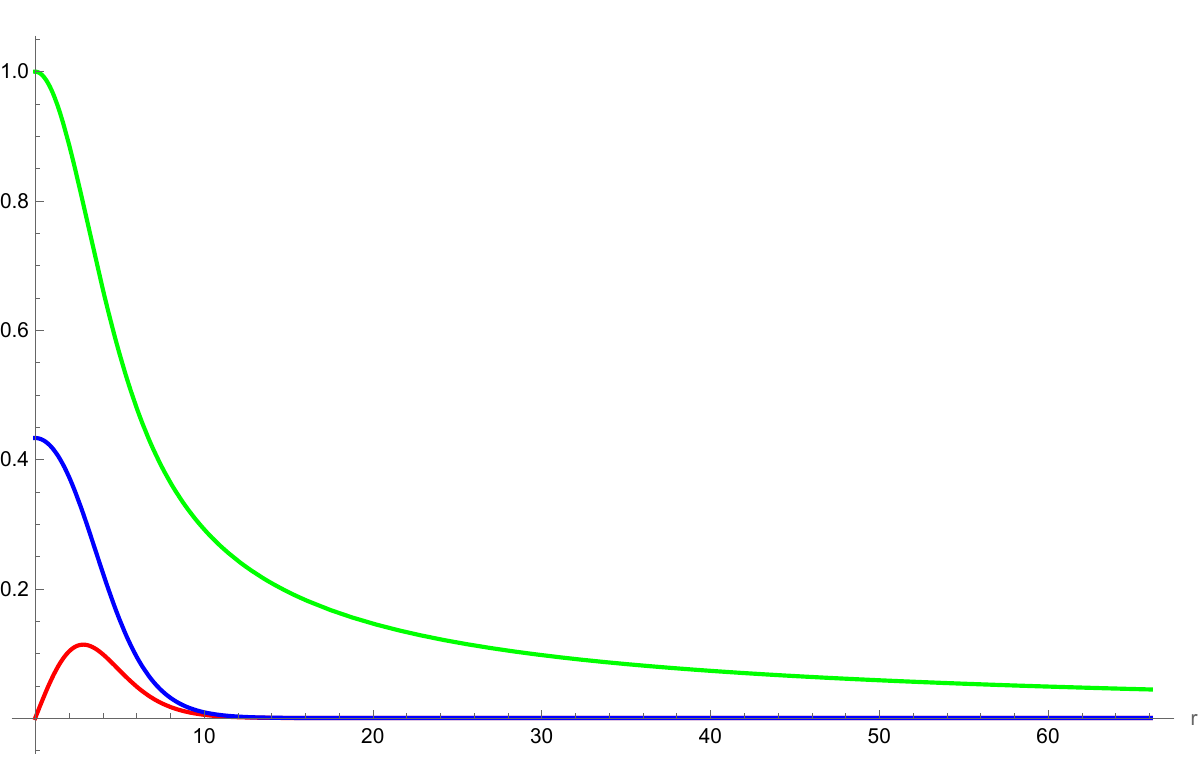}
\hfill\includegraphics[width=0.46\textwidth,height=140pt]{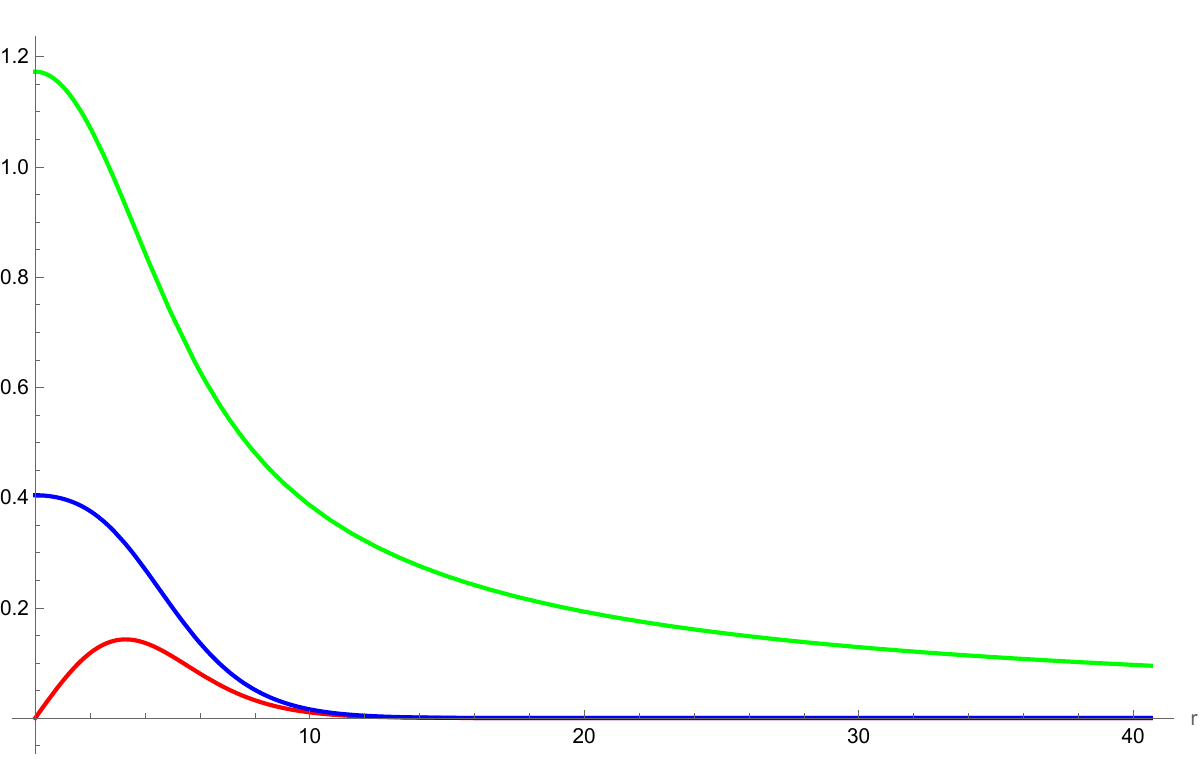}
\else
\noindent\includegraphics[width=0.46\textwidth,height=140pt]{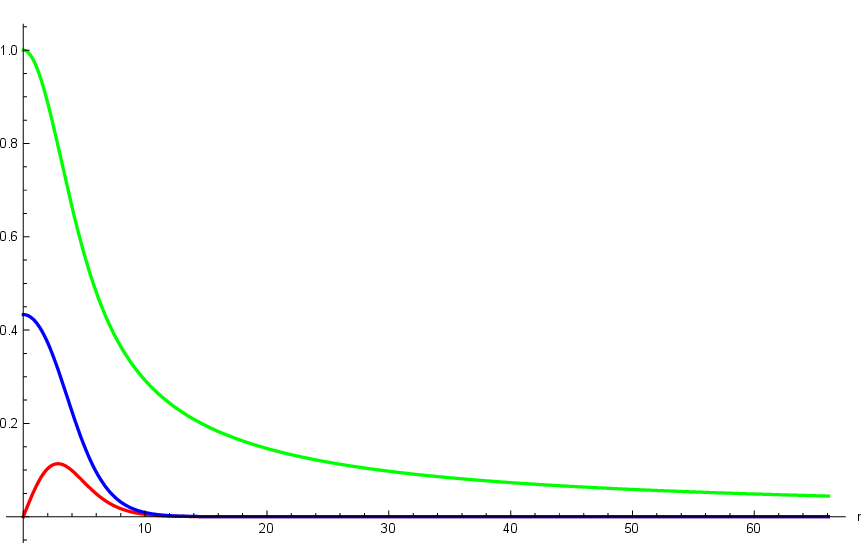}
\hfill\includegraphics[width=0.46\textwidth,height=140pt]{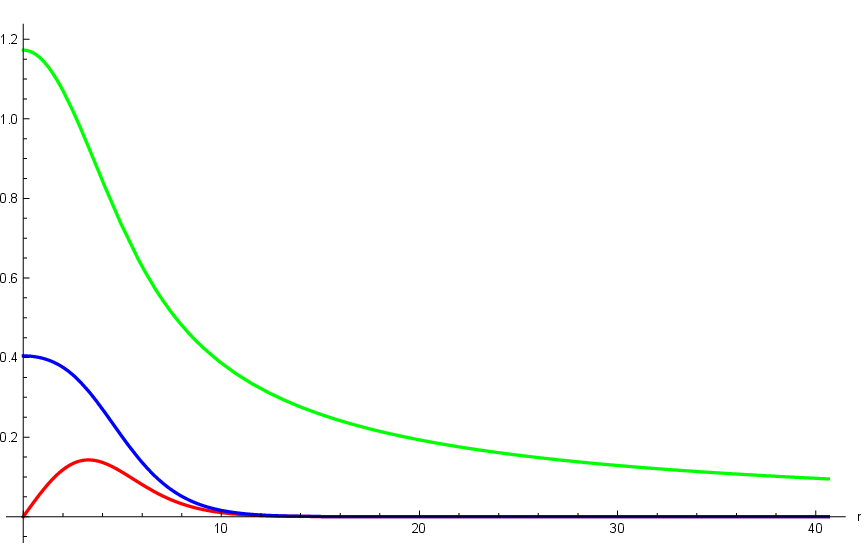}
\fi
 
\ifpdf
\noindent\includegraphics[width=0.46\textwidth,height=140pt]{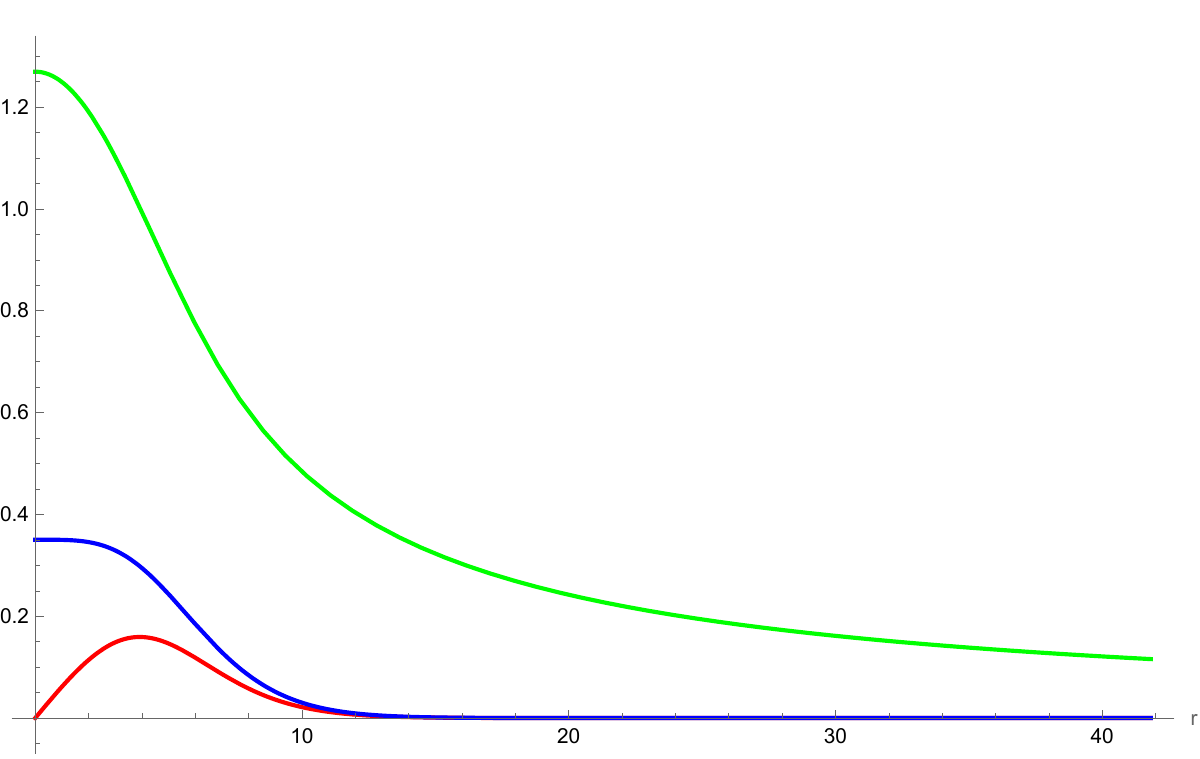}
\hfill\includegraphics[width=0.46\textwidth,height=140pt]{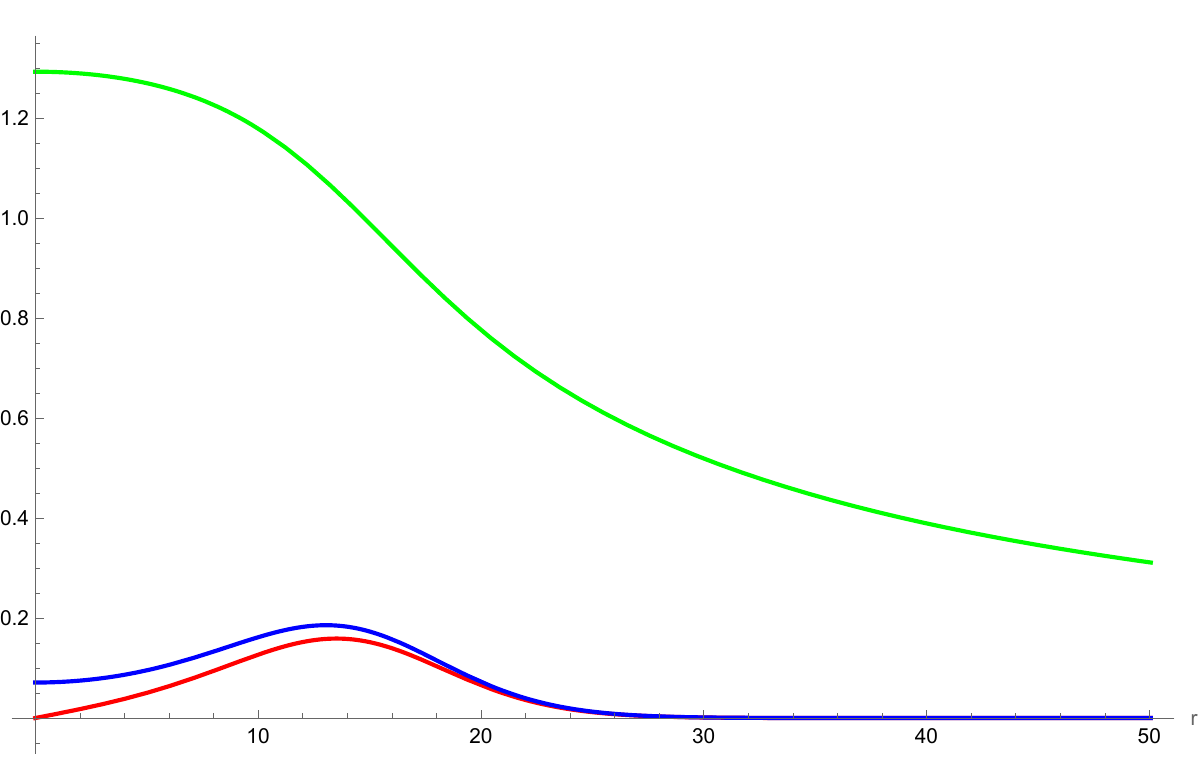}
\else
\noindent\includegraphics[width=0.46\textwidth,height=140pt]{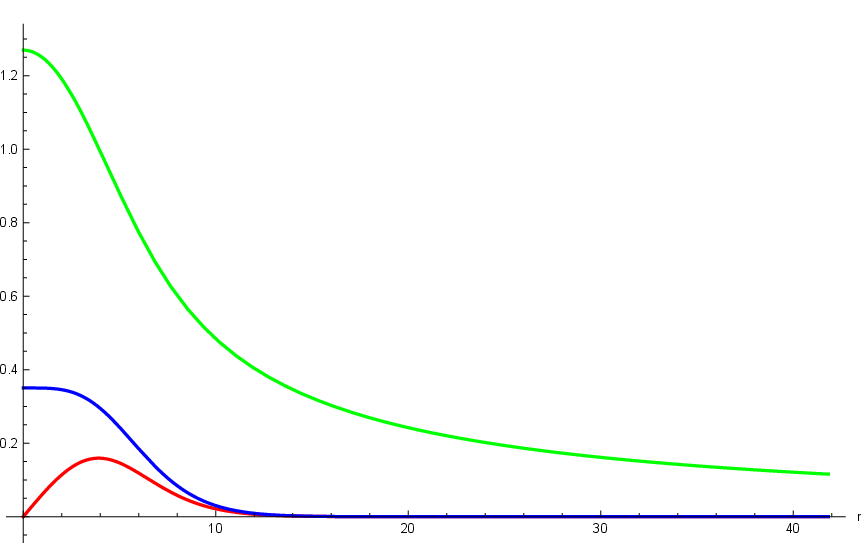}
\hfill\includegraphics[width=0.46\textwidth,height=140pt]{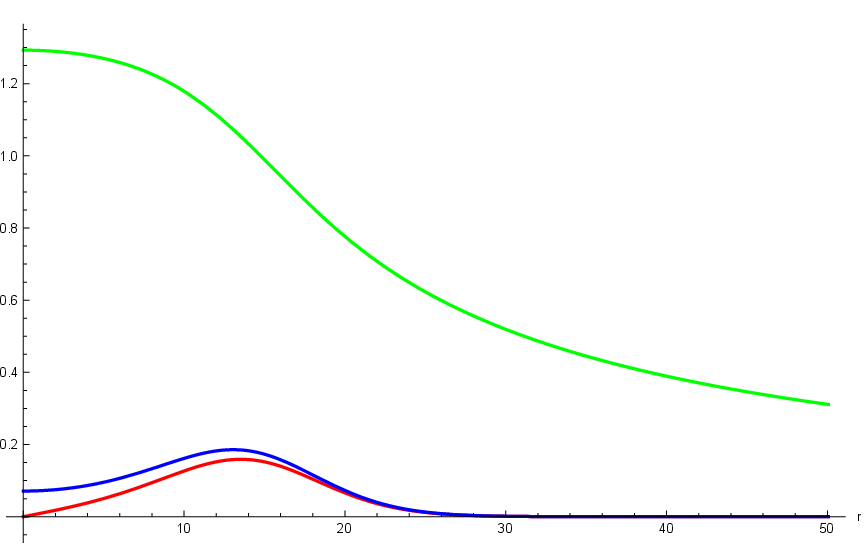}
\fi

\caption{
\footnotesize
3D,
$m=1$, $\coupling=1$, $M=0$;
the shooting method.
Profiles $v$ (blue), $u$ (red), and $h$ (green) of the solitary waves
corresponding to
${\omega}=0.903$ (top left),
${\omega}=0.640$ (top right),
${\omega}=0.457$ (center left),
${\omega}=0.345$ (center right),
${\omega}=0.269$ (bottom left),
and
$\omega=0.073$ (bottom right). The function $v$ is monotonic for $\omega\geq\omega_c\approx 0.269$ (this critical value corresponds to
$H(0)\approx1$; $H(0)=1$ being bottom left) and is otherwise increasing and then decreasing (as in bottom right).
We recall
that $\omega$
is obtained
from the
value $H(0)=s$
by computing it
from
$\omega_0$ and 
the effective mass $m_s$
(cf.~\eqref{def-ms})
so that it corresponds
to the chosen mass $m=1$;
see~\eqref{def-new-omega}.
}
    \label{Figure-profiles-3D-SimpleShootingMethod-M=0}
\end{figure}

\section{Discussion and
stability considerations}
\label{sect-stability}

Let us point out that
the two independent methods that we use -- the iterative method and the nested shooting method -- are in excellent agreement, collaborating one the other.
The iterative method works much faster
yielding the relative precision of
the order $10^{-4}$ or better
in just some five iterative steps.
On the negative side,
in 3D,
the subsequent repetition of iterations
does not seem to improve precision,
and moreover the iterative approach
breaks down for $\omega\lesssim 0.2m$.
The nested shooting method, while
taking longer,
allows us to improve the accuracy by
several orders of magnitude
for $\omega\lesssim 0.9m$
and also yields reliable
results for
small values of $\omega$
(albeit this demands delicate handling of
the intervals of initial data).

Let us now discuss which predictions
can already be made regarding stability
of DKG solitary waves.
One expects that
the stability of small amplitude solitary waves
(``weakly relativistic solitary waves'')
is inherited from the nonrelativistic limit of the model;
this was demonstrated in
\cite{comech2014linear,boussaid2019spectral}
for the NLD in any dimension
in the limit $\omega\lesssim m$.
As $\omega$ decreases,
critical points of
the charge $Q(\omega)$
(and also the energy $E(\omega)$
in view of the identity $E'(\omega)=\omega Q'(\omega)$)
and zero values of $E(\omega)$
indicate the collision of eigenvalues of the
linearization operators at the origin
(see \cite{berkolaiko2015vakhitov})
and hence
possibly the border of the stability region.
For example, the cubic NLD in 3D
has solitary waves
for $\omega\in(0,m)$
which are
linearly unstable in the
nonrelativistic limit
(the spectrum of the linearized operator
contains a pair of a positive and a negative
eigenvalues, just like the spectrum of 
cubic NLS in 3D),
and
according to \cite{PhysRevLett.116.214101}
at $\omega=0.936m$
-- the critical point where $Q'(\omega)$ vanishes -- 
this pair of real eigenvalues collides
at the origin and turns into a pair of
purely imaginary eigenvalues,
in agreement with
the Kolokolov stability criterion
\cite{kolokolov-1973}.
It was demonstrated numerically in
\cite{comech2025stable}
that indeed there is the onset
of spectral stability
below $\omega=0.936m$
(this situation was already predicted
via numerical simulations in
\cite{PhysRevLett.50.1230}).
Further, it was shown
in \cite{comech2025stable}
that the stability region
for cubic NLD in 3D
is
approximately
$0.254 m<\omega<0.936 m$,
before nonzero-real-part eigenvalues
start emerging by pairs or by quadruplets
from collisions of purely imaginary eigenvalues
which occur for some smaller values of $\omega$.

Since
3D DKG
in the limit $M/m\to+\infty$
of heavy scalar field
turns into cubic NLD,
we expect a very similar situation with the
spectral stability.
We expect to have linear instability
for $\omega$ near $m$, followed by
the collision at $z=0$
of eigenvalues of the linearization operator
at
$\omega_0$ which corresponds to
the minimum of $Q(\omega)$
(and also of $E(\omega)$
due to \eqref{dkg-dedq}), followed by
the onset of spectral stability
below $\omega_0$.
We have shown numerically
in Section~\ref{sect-dkg-numerics}
that
the critical value $\omega_0$
--
the minimum of $E(\omega)$
on Figure~\ref{fig-3d},
which corresponds to
the upper boundary of the stability region -- increases as $M$ decreases, with
$\omega_0\to m$ as $M\to 0$.
We do not observe vanishing of
$Q'(\omega)$ or $E(\omega)$
on $(0,\omega_0)$,
giving us the expectation that the
stability region could be rather
large,
like for
cubic NLD in 3D.
The numerical difficulties artifacts
that we encountered
both in the iterative method
(slow or absent
convergence of iterations)
and in the nested shooting method
(the need to accurately specify
the intervals for shooting parameters)
when $\omega$ decreases below
approximately $\omega=0.25 m$
suggest that the corresponding
solitary waves are unstable
in that region,
seemingly in agreement with the spectral instability observed in 3D NLD for these values of $\omega$ \cite{comech2025stable}.
At the same time,
though,
the likely linear instability for
$\omega$ near $m$
(when $E'(\omega)>0$)
does not affect
fast convergence of iterations.

Let us give a heuristic justification of
spectral stability of solitary waves
in the massless case $M=0$,
for $\omega\lesssim m$.
One can see from the second equation
in the system
\eqref{system-dkg-r}
(where we take $\coupling=1$)
that $\p_r v\approx -2m u$
(under the assumption that $h$ is much smaller
than $m$),
and then the first and the third
equations
from \eqref{system-dkg-r}
take the form
\begin{equation}\label{ttf}
    0=\frac{1}{2m}
\Bigl(\p_r^2+\frac{n-1}{r}\p_r\Bigr)
v
    -(m-\omega)v+hv,
    \qquad
-\Bigl(\p_r^2+\frac{n-1}{r}\p_r\Bigr)h
=v^2-u^2;
\end{equation}
we assume that $v$ and $u$ are real-valued.
To make sure that
the three terms in the right-hand
side of
the first equation
in \eqref{ttf}
are of the same order,
we consider $v$, $u$, $h$
as functions of $t=\epsilon r$,
$\epsilon:=\sqrt{m^2-\omega^2}$,
and define $H(t,\epsilon)$ by
$h(r,\epsilon)=\epsilon^2 H(\epsilon r,\epsilon)$,
so indeed $h$ will be much smaller
than $m$ if $H$ is bounded.
(We now consider
$v$, $u$, and $h$
as functions of $r$ and
$\epsilon=\sqrt{m^2-\omega^2}$
in place of $r$ and $\omega$.)
If we assume that
$\abs{v}\gg\abs{u}$,
then $\Delta h=v^2-u^2$
suggests that we introduce
$V(t,\epsilon)$
so that
\[
v(r,\epsilon)=\epsilon^2 V(\epsilon r,\epsilon)\,.
\]
It follows that
\[  
u(r,\epsilon)
\approx -\frac{1}{2m}\p_r v(r,\epsilon)
\approx
-\frac{\epsilon^3}{2m}
\p_t V(\epsilon r,\epsilon)\,,
\]
which is
indeed in agreement with our
assumption that
$u$ is much smaller than $v$,
and that
\[
\hat V(t)=\lim_{\epsilon\to 0}V(t,\epsilon)
\qquad \text{and} \qquad
\hat H(t)=\lim_{\epsilon\to 0}H(t,\epsilon)
\]
are to satisfy
the Schr\"odinger--Poisson system
\begin{equation}\label{sp}
-\frac{1}{2m}\hat V
=-\frac{1}{2m}
\Bigl(
\p_t^2+\frac{n-1}{t}\p_t\Bigr)
\hat V
-\hat H \hat V\,,
\qquad
-\Bigl(
\p_t^2+\frac{n-1}{t}\p_t\Bigr)
\hat H=\hat V^2\,,
\end{equation}
which is equivalent to the Choquard equation.
The rigorous construction of
$v(r,\omega)$ and $u(r,\omega)$
for $\omega$ near $m$ can be accomplished
precisely as in \cite{comech2018small}.
Since the solitary waves in
Schr\"odinger--Poisson system
\eqref{sp}
are orbitally stable
\cite{cazenave1982orbital},
one expects
that the spectral stability holds
for solitary wave solutions to
\eqref{dkg} with $M=0$ for $\omega$ close enough to $m$.

\section{Conclusion}

We constructed numerically
solitary waves
in Dirac--Klein--Gordon system in
1D and 3D
via the iteration procedure
and via the nested shooting method.
We showed that
for $\omega\gtrsim 0.5m$
the iterative procedure converges rather fast,
in just five
iterative steps,
with the relative error
of order $10^{-4}$
in the case of the massive
scalar field ($M=1$).
We were also able to implement
the nested shooting method
with respect to two parameters,
reducing the
relative error
by several orders of magnitude.
In the case of massless scalar field in 3D,
when the accuracy of the iterative
method is lacking,
we applied
the (standard, one-parameter) shooting method
which
becomes available in the case $M=0$.

We computed the dependencies
of $E$ and $Q$ on $\omega$.
Based on these dependencies
and the stability considerations in
\cite{berkolaiko2015vakhitov},
we expect that the
Dirac--Klein--Gordon solitary waves
are unstable for $\omega_0\le \omega<m$
and spectrally stable
for $\omega\lesssim \omega_0$,
with $\omega_0/m\to 0.936$ as $M/m\to+\infty$
(spinor field coupled to heavy scalar particles)
and $\omega_0/m\to 1$ as $M/m\to 0$ (coupling of the spinor field to very light or massless scalar particles).

\section*{Acknowledgements}

This work was supported by a grant from the Simons Foundation (851052, A.C.). JR is Fellow of \emph{Sistema Nacional de Investigadoras e Investigadores} and his research was partially supported by the DGAPA-UNAM's project PAPIIT IA102326.

\appendix

\section{Appendix:
the nonlinear Dirac equation}
\label{sect-nld}

The Soler model
\cite{jetp.8.260,PhysRevD.1.2766}
has the form
\begin{equation}\label{nld}
\jj\p\sb{t}\psi
=D_0\psi+(m-f(\psi^*\beta\psi))\beta\psi\,,
\qquad
\psi(t,x)\in\C^N,
\quad
x\in\R^n,
\end{equation}
where $D_0=-\jj\bm\alpha\cdot\nabla$;
the self-adjoint matrices
$\alpha^j$, $\beta$
satisfy \eqref{anti-alpha}
(for details about the nonlinear
Dirac equation
in higher dimensions, see e.g. \cite{boussaid2019nonlinear}).
The nonlinearity
is characterized by the function
$f\in C(\R,\R)\cap C^1(\R,\R\setminus\{0\})$,
$f(0)=0$.

We are interested in solitary wave solutions
$\phi(x,\omega)e^{-\jj\omega t}$
to \eqref{nld},
with $\phi$ satisfying
\begin{equation}\label{nld-stationary}
    \omega\phi=D_0\phi+\beta m\phi-f( \phi^*\beta\phi)\beta\phi\,.
\end{equation}
This stationary nonlinear Dirac equation
can be written in the form
\begin{equation}\label{nld-variational}
    \omega\delta_{\phi^*}Q = \delta_{\phi^*}E\,,
\end{equation}
with the energy functional
given by
\begin{equation}\label{def-E}
    E(\phi)
    =\int_{\R^n}
    \bigl(
    \phi^* D_0\phi
    +m \phi^*\beta\phi
    -F( \phi^*\beta\phi)
    \bigr)
    =K(\phi)+N(\phi)+V(\phi)\,,
\end{equation}
where
$K(\phi)$ and $N(\phi)$ are from
\eqref{def-kn}
and
\[
    V(\phi)
    =-\int_{\R^n}
    F( \phi^*\beta\phi)\,,
    \qquad
    F(\tau):=\int_0^\tau f(s)\,ds\,.
\]
We note that
for a $C^1$-family of solitary waves
$\phi(x,\omega)e^{-\jj\omega t}$,
$\omega\in\Omega\subset(-m,m)$,
one has
\[
    \p_\omega E(\phi) = \omega \p_\omega Q(\phi)\,,
    \qquad
    \omega\in\Omega\,.
\]
Indeed, using the relation
\eqref{nld-variational},
we derive (cf. \eqref{dkg-dedq}):
\[
    \p_\omega E(\phi)
    =
    2\Re\langle\delta_{\phi^*}E\,,
    \p_\omega
    \phi^*\rangle
    =2 \Re\langle\omega\delta_{\phi^*}Q,\p_\omega\phi^*\rangle
    =
    \omega\p_\omega Q(\phi)\,.
\]

If $\phi(x,\omega)e^{-\jj\omega t}$
is a solitary wave
with $\phi$ in the form \eqref{sw-1},
then by~\eqref{nld-stationary},
$v(r,\omega)$ and $u(r,\omega)$ satisfy the system
(see, e.g.,
\cite{esteban1995stationary,boussaid2017nonrelativistic})
\begin{equation}\label{nld-stationary-1}
    \begin{cases}
        \omega v=\p_r u+\frac{n-1}{r}u+(m-f(v^2-u^2))v, \\
        \omega u=-\p_r v-(m-f(v^2-u^2))u,
    \end{cases}
    \qquad
    r>0\,.
\end{equation}
In the case $n=1$,
the above system is considered
for $x\in\R$ instead of $r>0$;
cf. Remark~\ref{remark-nd-2}.

\subsubsection*{The virial identity
for the nonlinear Dirac equation}

Here is the virial identity for the
nonlinear Dirac equation
in $\R^n$, $n\ge 1$:
\begin{lemma}[
\mbox{\cite[Lemma IX.15]{boussaid2019nonlinear}}]
\label{Soler_Lemma_Virial_identity}

Let $F\in C^1(\R,\R)$ satisfy
$F(0)=0$.
Let $n\in\N$, $N\in\N$, $\omega\in\R$,
and assume that $\phi\in L\sp\infty\sb{\mathrm{loc}}(\R^n,\C^N)$ satisfies
\[
    \omega\phi=D_0\phi
    +m\beta\phi
    -F'( \phi^*\beta\phi)\beta\phi
\]
in the sense of distributions.
Assume furthermore that
$\phi\in H^1(\R^n,\C^N)$ and
$G( \phi^*\beta\phi)\in L^1(\R^n)$.
Then $\phi$ satisfies the \emph{virial identity}
\[
    \omega\int_{\R^n}\phi^*\phi
    =\frac{n-1}{n}\int_{\R^n}\phi^* D_0\phi
    +m\int_{\R^n}  \phi^*\beta\phi
    -\int_{\R^n} F( \phi^*\beta\phi)\,.
\]
\end{lemma}

A non-rigorous way to derive
this virial identity is
following \cite{derrick1964comments}:
define
\[
    \phi_\lambda(x,\omega)=\phi(x/\lambda,\omega)\,,
    \qquad
    x\in\R^n,
    \quad\omega\in(0,m)\,,
    \quad
    \lambda>0\,;
\]
coupling \eqref{nld-variational}
with $\p_\lambda\phi_\lambda^*$
(and
adding to the result its complex conjugate)
we arrive at
\begin{equation}\label{nld-ty}
    \omega
    \p_\lambda\at{\lambda=1}Q(\phi_\lambda)
    =
    \p_\lambda\at{\lambda=1}E(\phi_\lambda)
    =
    \p_\lambda\at{\lambda=1}
    \bigl(
    K(\phi_\lambda)
    +N(\phi_\lambda)
    +V(\phi_\lambda)
    \bigr).
\end{equation}
Taking into account that
\begin{equation}\label{nld-scaling}
    Q(\phi_\lambda)=\lambda^n Q(\phi)\,,
    \quad
    K(\phi_\lambda)=\lambda^{n-1} K(\phi)\,,
    \quad
    N(\phi_\lambda)=\lambda^n N(\phi)\,,
    \quad
    V(\phi_\lambda)=\lambda^n V(\phi)\,,
\end{equation}
we derive from~\eqref{nld-ty}
the \emph{virial identity} stated in Lemma~\ref{Soler_Lemma_Virial_identity}:
\begin{equation}\label{nld-virial-0}
    \omega Q(\phi) = \frac{n-1}{n} K(\phi) + N(\phi) + V(\phi) = E(\phi) - \frac{1}{n} K(\phi)\,.
\end{equation}

\begin{remark}\label{remark-nld-simple}
We note that, multiplying~\eqref{nld-stationary} by $\phi^*$ and integrating, one gets the relation
\begin{equation}\label{nld-easy}
    \omega Q(\phi)
    =K(\phi)+N(\phi)
    -\int_{\R^n}\phi^*\beta\phi\,f( \phi^*\beta\phi)\,.
\end{equation}
In the cubic case, $f(\tau)=\tau$ hence
$F(\tau)=\tau^2/2$,
\eqref{nld-easy}
takes the form
\[
    \omega Q(\phi)=K(\phi)+N(\phi)+2V(\phi)\,.
\]
The above and
\eqref{nld-virial-0} yield
$V(\phi)=-\frac{1}{n}K(\phi)$;
hence, in the cubic case,
\eqref{def-E} can be written as
\[
    E(\phi)=
    K(\phi)+N(\phi)+V(\phi)
    =N(\phi) - (n-1) V(\phi)\,,
\]
while the virial identity
\eqref{nld-virial-0}
takes the form
\begin{equation}\label{nld-virial-cubic}
\omega Q(\phi)=\frac{n-1}{n}K(\phi)+N(\phi)+V(\phi) = \frac{n-2}{n}K(\phi)+N(\phi) = N(\phi) - (n-2) V(\phi)\,.
\end{equation}
To check the accuracy of the numerics, we define
the error as the difference of the
left and right hand sides of the virial identity,
\begin{align}
\label{def-varepsilon-nld}
\varepsilon
=\omega Q(\phi)-\frac{n-1}{n}K(\phi)-N(\phi)-V(\phi).
\end{align}
\end{remark}

\section{Appendix:
Dirac--Klein--Gordon
system in 1D}
\label{sect_appendix_1D}

While the 1D case is unphysical,
it always serves as a convenient
playground for testing different approaches;
this has been the case for the
1D NLD, for both the
Gross--Neveu model (scalar self-interaction)
and massive Thirring model
(vector self-interaction);
we refer to the review of the subject in \cite{boussaid2019nonlinear}
and also to \cite{aldunate2023results}
for more recent results.
For this purpose,
we present here the energy
of solitary waves in 1D DKG
as a function of $\omega$.
We point out that
the energy $E(\omega)$
is positive and monotonically decreasing,
without zeros and critical points
This indicates that, just like
in the case of the cubic 1D NLD
considered in
\cite{berkolaiko2012spectral,berkolaiko2015vakhitov},
there are no collisions of eigenvalues
of the linearized equation at $z=0$
and one may expect (spectral) stability
of solitary waves
with all values of $\omega$.

\begin{figure}[!ht]
\begin{center}
\ifpdf
\noindent\includegraphics[width=0.44\textwidth,height=140pt]{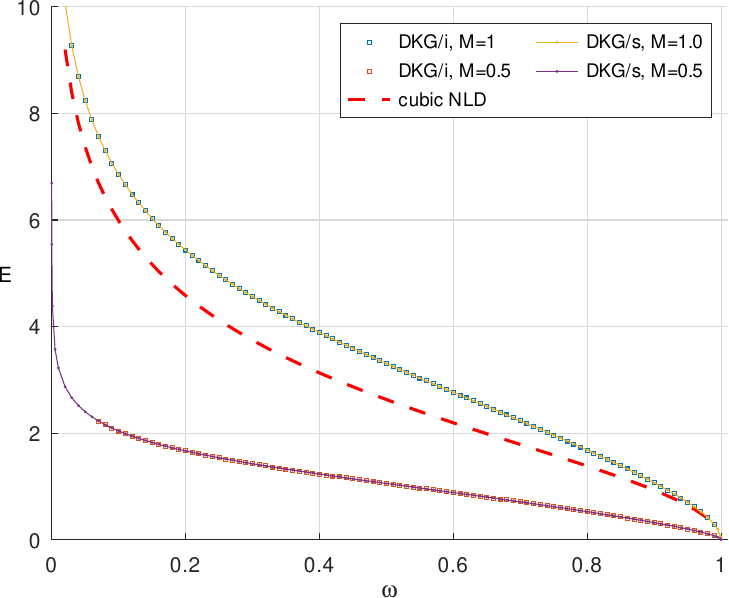}
\else
\noindent\includegraphics[width=0.44\textwidth,height=140pt]{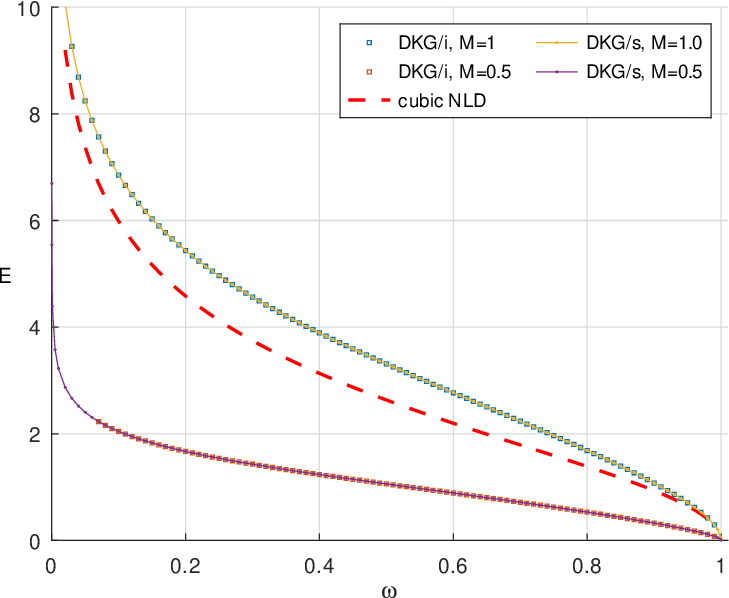}
\fi
\end{center}
\caption{\footnotesize
1D; $m=1$,
$\coupling=1$, $M=1$ and $M=1/2$.
Energy $E$
of solitary waves
as a functions of $\omega$
for the Dirac--Klein--Gordon system
computed by iterative and nested shooting methods.
The squares (labeled ``DKG/i'') correspond to the iterative method and the dots (labeled ``DKG/s'') to the nested shooting method.
The dashed line corresponds to the solitary waves of the cubic NLD.
}
\label{fig-1d}
\end{figure}
\begin{figure}[!ht]

\ifpdf
\noindent\includegraphics[width=0.44\textwidth,height=140pt]{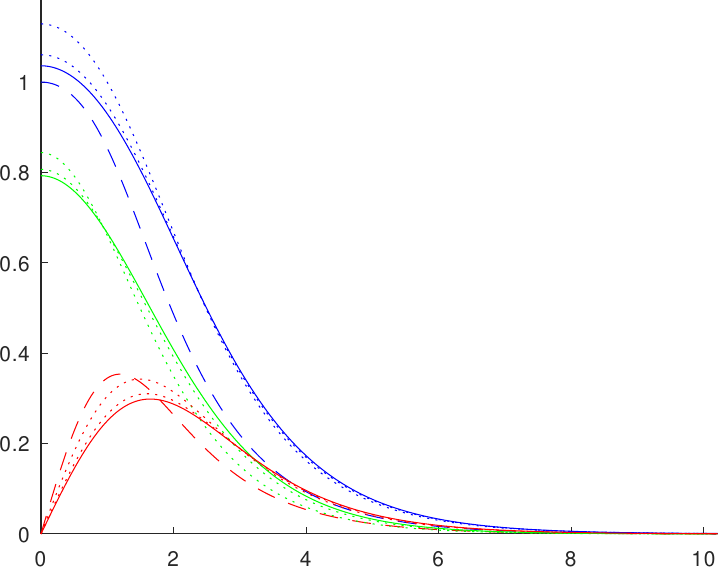}
\hfill\includegraphics[width=0.44\textwidth,height=140pt]{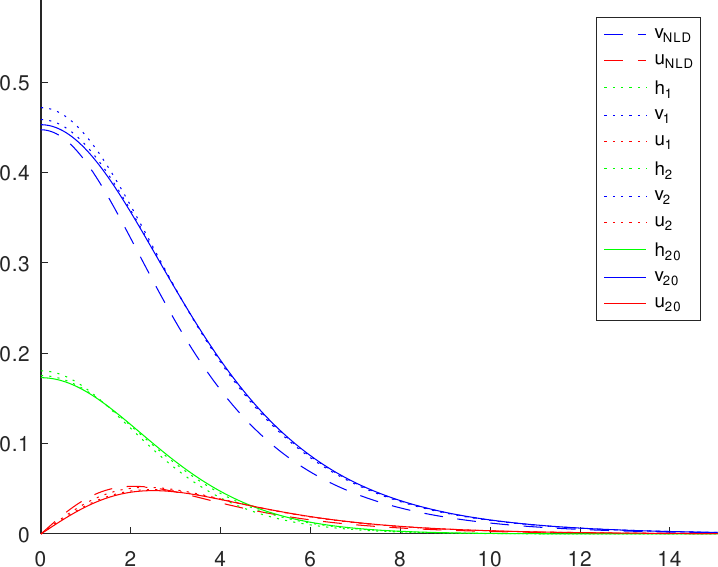}
\else
\noindent\includegraphics[width=0.44\textwidth,height=140pt]{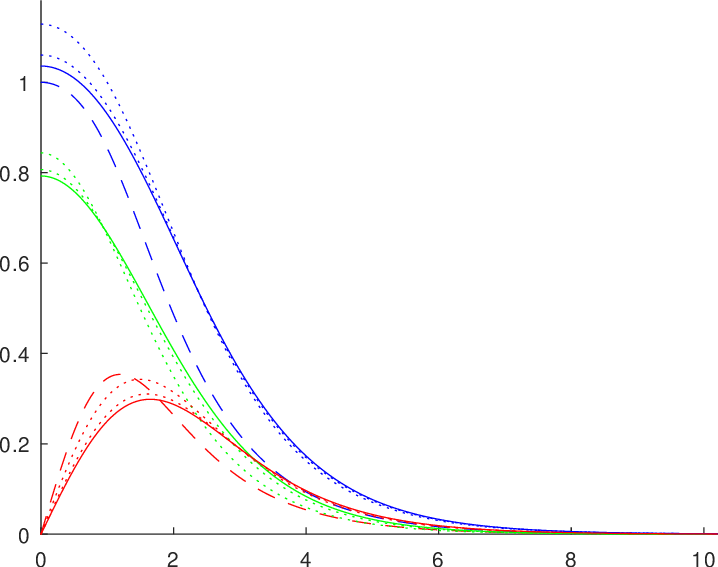}
\hfill\includegraphics[width=0.44\textwidth,height=140pt]{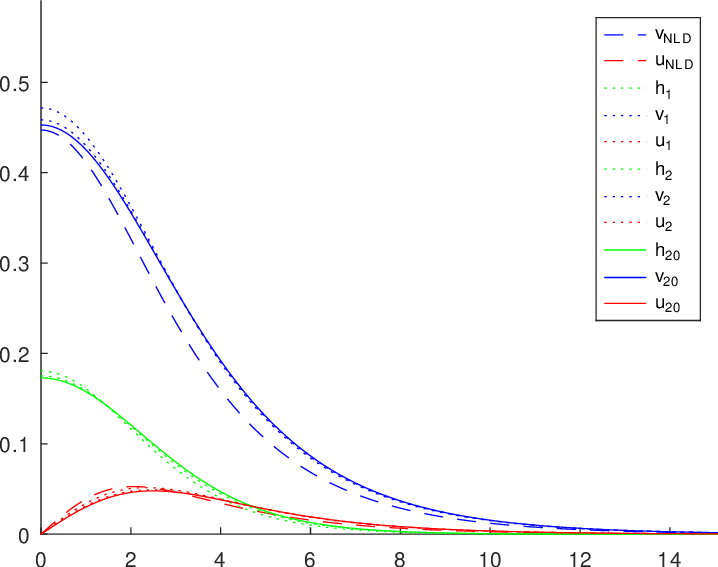}
\fi
\caption{\footnotesize
1D, the iterative method;
$m=1$,
$\coupling=1$, $M=1$;
$\omega=0.5$ (left)
and
$\omega=0.9$
(right).
Solitary waves of cubic NLD
(dashed)
and iterations
of solitary waves of
DKG system:
first and second (dotted)
and
twentieth
(solid lines)
iterations are plotted.
}
\label{fig-1-1.00}
\end{figure}

\clearpage

\section{Appendix: accuracy results}

\subsubsection*{Accuracy in 3D case}

\label{sect_appendix_Tables_3D}

\begin{table}[!ht]
\begin{center}
\caption{\footnotesize
3D, $m=1$, $\coupling=1$, $M=1$.
Accuracy of several steps
from the iterative method
and of the nested shooting method.
}
\scriptsize\begin{tabular}{ccccccc} \toprule
$\omega$ & step & $v(0)$ & $h(0)$ & $Q$ & $E$ & $|\varepsilon|/(\omega Q)$\\
\midrule
\multirow{8}{*}{$0.2$} & $\#\,000$ & $1.202271$ &  & $10695.4$ & $2934.38$ & $1.4\times10^{-2}$ \\
 & $\#\,001$ & $1.305827$ & $1.709858$ & $23692.2$ & $6406.64$ & $1.3\times10^{-1}$ \\
 & $\#\,002$ & $0.988337$ & $1.235453$ & $5222.40$ & $1508.87$ & $3.1\times10^{-1}$ \\
 & $\#\,006$ & $0.736237$ & $0.958079$ & $2247.44$ & $671.769$ & $1.0$ \\
 & $\#\,020$ & $0.590474$ & $0.803350$ & $1350.15$ & $413.161$ & $1.9$ \\
 & $\#\,060$ & $0.596420$ & $0.808268$ & $1363.56$ & $418.591$ & $1.9$ \\
 & $\#\,200$ & $0.592860$ & $0.805379$ & $1364.31$ & $417.104$ & $1.9$ \\
\cdashline{2-7}\noalign{\vskip 2pt}
&nested shootings & $1.192767$ & $1.515806$ & $11923.9$ & $3338.23$ & $1.4\times10^{-8}$ \\
\midrule
\multirow{8}{*}{$0.3$} & $\#\,000$ & $1.283529$ &  & $2491.12$ & $1019.14$ & $3.6\times10^{-4}$ \\
 & $\#\,001$ & $1.455864$ & $1.903212$ & $3829.79$ & $1588.62$ & $6.5\times10^{-2}$ \\
 & $\#\,002$ & $1.295538$ & $1.725898$ & $2410.60$ & $1019.31$ & $4.1\times10^{-2}$ \\
 & $\#\,006$ & $1.358431$ & $1.788225$ & $2815.69$ & $1187.02$ & $3.3\times10^{-3}$ \\
 & $\#\,020$ & $1.365813$ & $1.795567$ & $2869.09$ & $1208.90$ & $8.2\times10^{-4}$ \\
 & $\#\,060$ & $1.365813$ & $1.795540$ & $2869.05$ & $1208.88$ & $5.7\times10^{-4}$ \\
 & $\#\,200$ & $1.365813$ & $1.795521$ & $2868.80$ & $1208.79$ & $5.0\times10^{-4}$ \\
\cdashline{2-7}\noalign{\vskip 2pt}
&nested shootings & $1.365817$ & $1.795538$ & $2869.96$ & $1209.16$ & $7.2\times10^{-8}$ \\
\midrule
\multirow{8}{*}{$0.5$} & $\#\,000$ & $1.380579$ &  & $384.296$ & $254.890$ & $2.9\times10^{-4}$ \\
 & $\#\,001$ & $1.716580$ & $1.850420$ & $522.358$ & $357.012$ & $1.6\times10^{-2}$ \\
 & $\#\,002$ & $1.718751$ & $1.819921$ & $511.744$ & $351.638$ & $3.6\times10^{-3}$ \\
 & $\#\,006$ & $1.721600$ & $1.813302$ & $511.227$ & $351.656$ & $4.5\times10^{-4}$ \\
 & $\#\,020$ & $1.721604$ & $1.813295$ & $511.226$ & $351.656$ & $4.4\times10^{-4}$ \\
 & $\#\,060$ & $1.721604$ & $1.813295$ & $511.226$ & $351.656$ & $4.4\times10^{-4}$ \\
 & $\#\,200$ & $1.721604$ & $1.813295$ & $511.226$ & $351.656$ & $4.4\times10^{-4}$ \\
\cdashline{2-7}\noalign{\vskip 2pt}
&nested shootings & $1.721576$ & $1.813300$ & $511.479$ & $351.784$ & $4.9\times10^{-8}$ \\
\midrule
\multirow{8}{*}{$0.9$} & $\#\,000$ & $1.065072$ &  & $48.1815$ & $49.2058$ & $1.0\times10^{-4}$ \\
 & $\#\,001$ & $1.369138$ & $0.762320$ & $112.429$ & $113.385$ & $2.5\times10^{-2}$ \\
 & $\#\,002$ & $1.192838$ & $0.660407$ & $97.4420$ & $97.8913$ & $9.7\times10^{-3}$ \\
 & $\#\,006$ & $1.102409$ & $0.609214$ & $90.1148$ & $90.3264$ & $4.3\times10^{-4}$ \\
 & $\#\,020$ & $1.100615$ & $0.608188$ & $89.9879$ & $90.1938$ & $2.1\times10^{-4}$ \\
 & $\#\,060$ & $1.100615$ & $0.608188$ & $89.9879$ & $90.1938$ & $2.1\times10^{-4}$ \\
 & $\#\,200$ & $1.100615$ & $0.608188$ & $89.9879$ & $90.1938$ & $2.1\times10^{-4}$ \\
\cdashline{2-7}\noalign{\vskip 2pt}
&nested shootings & $1.100603$ & $0.608187$ & $90.1073$ & $90.3018$ & $1.0\times10^{-7}$ \\
\midrule
\multirow{8}{*}{$0.99$} & $\#\,000$ & $0.419441$ &  & $73.4064$ & $74.0519$ & $1.6\times10^{-5}$ \\
 & $\#\,001$ & $0.450900$ & $0.140742$ & $106.708$ & $107.489$ & $2.3\times10^{-3}$ \\
 & $\#\,002$ & $0.416766$ & $0.126528$ & $101.854$ & $102.531$ & $1.1\times10^{-3}$ \\
 & $\#\,006$ & $0.388921$ & $0.114851$ & $98.4323$ & $99.0284$ & $9.9\times10^{-5}$ \\
 & $\#\,020$ & $0.387057$ & $0.114066$ & $98.2297$ & $98.8205$ & $2.6\times10^{-5}$ \\
 & $\#\,060$ & $0.387057$ & $0.114065$ & $98.2297$ & $98.8204$ & $2.6\times10^{-5}$ \\
 & $\#\,200$ & $0.387057$ & $0.114065$ & $98.2297$ & $98.8204$ & $2.6\times10^{-5}$ \\
\cdashline{2-7}\noalign{\vskip 2pt}
&nested shootings & $0.387052$ & $0.114063$ & $98.3774$ & $98.9668$ & $1.6\times10^{-7}$ \\
\bottomrule
\end{tabular}
\label{table-3d-iterations}
\end{center}
\end{table}

We present the results for more values
of $\omega$ and $M$
in Tables~\ref{supp-table-3D-NestedshootingMethod-M=1}--\ref{supp-table-3D-IterativeMethod-M=0} in the Supplementary Material.

\clearpage

\subsubsection*{Accuracy in 1D case}

\begin{table}[!ht]
\begin{center}
\caption{\footnotesize
1D, $m=1$, $\coupling=1$, $M=1$.
Accuracy of several steps
from the iterative method
and of the nested shooting method.
}
\scriptsize\begin{tabular}{ccccccc} \toprule
$\omega$ & step & $v(0)$ & $h(0)$ & $Q$ & $E$ & $|\varepsilon|/(\omega Q)$\\
\midrule
\multirow{8}{*}{$0.1$} & $\#\,000$ & $1.341831$ &  & $19.9045$ & $5.98588$ & $5.3\times10^{-4}$ \\
 & $\#\,001$ & $1.474526$ & $1.640793$ & $24.5375$ & $7.65595$ & $6.1\times10^{-1}$ \\
 & $\#\,002$ & $1.405806$ & $1.566962$ & $22.5392$ & $7.14473$ & $2.5\times10^{-1}$ \\
 & $\#\,006$ & $1.365587$ & $1.523180$ & $21.4194$ & $6.85923$ & $6.5\times10^{-3}$ \\
 & $\#\,020$ & $1.364582$ & $1.522055$ & $21.3920$ & $6.85235$ & $4.5\times10^{-5}$ \\
 & $\#\,060$ & $1.364580$ & $1.522055$ & $21.3921$ & $6.85235$ & $2.2\times10^{-5}$ \\
 & $\#\,200$ & $1.364580$ & $1.522054$ & $21.3920$ & $6.85235$ & $9.7\times10^{-6}$ \\
\cdashline{2-7}\noalign{\vskip 2pt}
&nested shootings & $1.364582$ & $1.522072$ & $21.3944$ & $6.85258$ & $6.0\times10^{-7}$ \\
\midrule
\multirow{8}{*}{$0.3$} & $\#\,000$ & $1.183282$ &  & $6.35816$ & $3.74621$ & $2.8\times10^{-4}$ \\
 & $\#\,001$ & $1.334396$ & $1.200059$ & $8.74198$ & $5.20707$ & $3.2\times10^{-1}$ \\
 & $\#\,002$ & $1.253988$ & $1.144663$ & $7.92933$ & $4.74011$ & $1.0\times10^{-1}$ \\
 & $\#\,006$ & $1.221892$ & $1.121517$ & $7.62422$ & $4.56516$ & $1.1\times10^{-3}$ \\
 & $\#\,020$ & $1.221559$ & $1.121270$ & $7.62113$ & $4.56340$ & $5.1\times10^{-6}$ \\
 & $\#\,060$ & $1.221559$ & $1.121270$ & $7.62113$ & $4.56340$ & $5.1\times10^{-6}$ \\
 & $\#\,200$ & $1.221559$ & $1.121270$ & $7.62113$ & $4.56340$ & $6.5\times10^{-6}$ \\
\cdashline{2-7}\noalign{\vskip 2pt}
&nested shootings & $1.221557$ & $1.121280$ & $7.62235$ & $4.56381$ & $5.1\times10^{-8}$ \\
\midrule
\multirow{8}{*}{$0.5$} & $\#\,000$ & $1.000045$ &  & $3.46351$ & $2.63336$ & $4.9\times10^{-5}$ \\
 & $\#\,001$ & $1.128597$ & $0.843935$ & $4.94971$ & $3.74646$ & $1.6\times10^{-1}$ \\
 & $\#\,002$ & $1.060410$ & $0.806811$ & $4.52115$ & $3.41907$ & $4.7\times10^{-2}$ \\
 & $\#\,006$ & $1.036194$ & $0.792903$ & $4.37843$ & $3.31004$ & $3.7\times10^{-4}$ \\
 & $\#\,020$ & $1.036009$ & $0.792794$ & $4.37737$ & $3.30923$ & $1.1\times10^{-5}$ \\
 & $\#\,060$ & $1.036009$ & $0.792794$ & $4.37737$ & $3.30923$ & $9.8\times10^{-6}$ \\
 & $\#\,200$ & $1.036009$ & $0.792794$ & $4.37737$ & $3.30923$ & $1.1\times10^{-5}$ \\
\cdashline{2-7}\noalign{\vskip 2pt}
&nested shootings & $1.036008$ & $0.792800$ & $4.37814$ & $3.30966$ & $3.6\times10^{-8}$ \\
\midrule
\multirow{8}{*}{$0.9$} & $\#\,000$ & $0.447236$ &  & $0.96843$ & $0.93408$ & $1.1\times10^{-6}$ \\
 & $\#\,001$ & $0.471636$ & $0.180498$ & $1.16825$ & $1.12456$ & $1.2\times10^{-2}$ \\
 & $\#\,002$ & $0.458404$ & $0.175210$ & $1.13124$ & $1.08834$ & $3.6\times10^{-3}$ \\
 & $\#\,006$ & $0.452861$ & $0.172934$ & $1.11660$ & $1.07400$ & $3.5\times10^{-5}$ \\
 & $\#\,020$ & $0.452804$ & $0.172910$ & $1.11644$ & $1.07385$ & $2.0\times10^{-6}$ \\
 & $\#\,060$ & $0.452804$ & $0.172910$ & $1.11644$ & $1.07385$ & $2.0\times10^{-6}$ \\
 & $\#\,200$ & $0.452804$ & $0.172910$ & $1.11644$ & $1.07385$ & $2.0\times10^{-6}$ \\
\cdashline{2-7}\noalign{\vskip 2pt}
&nested shootings & $0.452805$ & $0.172911$ & $1.11666$ & $1.07405$ & $4.4\times10^{-8}$ \\
\midrule
\multirow{8}{*}{$0.99$} & $\#\,000$ & $0.141429$ &  & $0.28492$ & $0.28396$ & $2.4\times10^{-7}$ \\
 & $\#\,001$ & $0.142494$ & $0.019599$ & $0.29341$ & $0.29242$ & $1.9\times10^{-4}$ \\
 & $\#\,002$ & $0.141870$ & $0.019451$ & $0.29242$ & $0.29142$ & $7.2\times10^{-5}$ \\
 & $\#\,006$ & $0.141505$ & $0.019365$ & $0.29185$ & $0.29085$ & $1.4\times10^{-6}$ \\
 & $\#\,020$ & $0.141497$ & $0.019363$ & $0.29184$ & $0.29084$ & $9.8\times10^{-8}$ \\
 & $\#\,060$ & $0.141497$ & $0.019363$ & $0.29184$ & $0.29084$ & $9.8\times10^{-8}$ \\
 & $\#\,200$ & $0.141497$ & $0.019363$ & $0.29184$ & $0.29084$ & $9.8\times10^{-8}$ \\
\cdashline{2-7}\noalign{\vskip 2pt}
&nested shootings & $0.141497$ & $0.019363$ & $0.29189$ & $0.29090$ & $3.3\times10^{-9}$ \\
\bottomrule
\end{tabular}
\label{table-1d-iterations}
\end{center}
\end{table}

We present the
results for more values of $\omega$ and $M$
in Tables~\ref{supp-table-1D-NestedshootingMethod-M=1}--\ref{supp-table-1D-IterativeMethod-M=0.50} in the Supplementary Material.

\clearpage
\bibliographystyle{sima-doi}
\bibliography{bibcomech}
\end{document}


\renewcommand{\theequation}{\thesection.\arabic{equation}}
\newcommand{\sect}[1]{\setcounter{equation}{0}\section{#1}}

\title{Numerical study
of solitary waves\\in
Dirac--Klein--Gordon system\\
Supplementary Material
}

\author{
Andrew Comech\,\orcidlink{0000-0003-2112-9359}
\\
{\small\it
Mathematics Department,
Texas A\&M University, College Station, TX 77843, USA}
\\[1ex]
Julien Ricaud\,\orcidlink{0000-0002-6777-7469}
\\
\small\it
Departamento de F\'{i}sica Matem\'{a}tica,
Instituto de Investigaciones en Matem\'aticas Aplicadas y en Sistemas,
\\
\small\it
Universidad Nacional Aut\'onoma de M\'exico, CP 04510, Ciudad de M\'exico, M\'exico
\\[1ex]
Marco Roque
\\
\small\it
Caver College of Aviation, 
Science, and Nursing,
Henderson State University, Arkadelphia, AR 71923
\\
\small\it
Arkansas State University,
State University, AR 72467
}

\maketitle

This supplementary material collects values of the
quantities computed by the nested shooting method
for $M>0$ and (simple) shooting method for $M=0$. Figures \ref{fig-3d} and \ref{fig-1d} (in the main manuscript) have been plotted based on the values of $E(\omega)$ from these tables.

We refer the reader to
~\eqref{def-Q},
\eqref{dkg-energy-no-derivatives}, and 
\eqref{for-epsilon}--\eqref{dkg-virial-noderivatives-with-E}
for the definitions and formulae used 
to compute
$Q$, $E$, and the relative error $|\varepsilon|/(\omega Q)$.

\renewcommand\thetable{S\arabic{table}}\setcounter{table}{0}
\section*{3D case}
\subsection*{The nested shooting method}
\begin{remark*}
    In Table~\ref{supp-table-3D-NestedshootingMethod-M=1}, the values for $\omega=0.16$ and $\omega=0.17$ are presented as examples of the method not converging systematically for small $\omega$'s unless we fine-tune the initial values $v_0^\pm$ and $h_0^\pm$ (see their definitions in the subsection ``The nested shooting method'' of Section~\ref{sect-dkg-numerics} in the main manuscript); see Remark~\ref{remark-small-omega} in the main manuscript.
    In order to obtain the values for $\omega\in [0.18, 0.24]$ we had to fine-tune our choice of initial values $v_0^\pm$ and $h_0^\pm$ so that converges to a solitary wave. We based our choice for this hand-chosen initial values on $(v_\omega(0), h_\omega(0))$ obtained previously for the next $\omega$'s:. E.g., for $\omega=0.24$ we set $(v_0^-, v_0^+) = (1.24, 1.27)$ and $(h_0^-, h_0^+) = (1.62, 1.65)$, based on the $(v_\omega(0), h_\omega(0))$ for $\omega=0.25$ and $0.26$.
\end{remark*}
\begin{longtable}{lcccccccc}
\caption{
    \footnotesize
    3D, the nested shooting method; $m=1$, $\coupling=1$, and $M=1$.
}
\label{supp-table-3D-NestedshootingMethod-M=1} \\
\toprule
$\omega$ & $v(0)$ & $h(0)$ & $Q$ & $E$ & $|\varepsilon|/(\omega Q)$ \\
\midrule
\endfirsthead 

\multicolumn{2}{l}{\tablename\ \thetable\ } & \multicolumn{4}{r}{\emph{Continued from previous page}} \\ 
\hline
$\omega$ & $v(0)$ & $h(0)$ & $Q$ & $E$ & $|\varepsilon|/(\omega Q)$ \\
\midrule
\endhead 

\cdashline{1-6}
\multicolumn{6}{r}{\emph{Continued on next page}} \\ 
\endfoot

\hline
\multicolumn{6}{r}{\emph{End of \tablename\ \thetable}} \\ 
\endlastfoot

$0.16$ & $1.4\times10^{-26}$ & $1.740000$ & $8.8\times10^{11}$ & $1.3\times10^{6}$ & $1.0$ \\
$0.17$ & $1.158750$ & $1.414869$ & $20596.7$ & $5080.06$ & $4.3\times10^{-2}$ \\
$0.18$ & $1.168827$ & $1.448243$ & $17322.2$ & $4359.57$ & $1.1\times10^{-9}$ \\
$0.19$ & $1.180329$ & $1.482007$ & $14299.3$ & $3800.96$ & $1.2\times10^{-8}$ \\
$0.2$ & $1.192767$ & $1.515806$ & $11923.9$ & $3338.23$ & $1.4\times10^{-8}$ \\
$0.21$ & $1.206179$ & $1.549296$ & $10034.3$ & $2951.20$ & $9.3\times10^{-9}$ \\
$0.22$ & $1.220577$ & $1.582161$ & $8514.26$ & $2624.67$ & $2.4\times10^{-8}$ \\
$0.23$ & $1.235947$ & $1.614113$ & $7279.17$ & $2346.98$ & $5.8\times10^{-8}$ \\
\end{longtable}

\begin{longtable}{lcccccccc}
\caption{
    \footnotesize
    3D, the nested shooting method; $m=1$, $\coupling=1$, and $M=1/2$.
}
\label{supp-table-3D-NestedshootingMethod-M=0.5} \\
\toprule
$\omega$ & $v(0)$ & $h(0)$ & $Q$ & $E$ & $|\varepsilon|/(\omega Q)$ \\
\midrule
\endfirsthead 

\multicolumn{2}{l}{\tablename\ \thetable\ } & \multicolumn{4}{r}{\emph{Continued from previous page}} \\ 
\hline
$\omega$ & $v(0)$ & $h(0)$ & $Q$ & $E$ & $|\varepsilon|/(\omega Q)$ \\
\midrule
\endhead 

\cdashline{1-6}
\multicolumn{6}{r}{\emph{Continued on next page}} \\ 
\endfoot

\hline
\multicolumn{6}{r}{\emph{End of \tablename\ \thetable}} \\ 
\endlastfoot

$0.2$ & $0.633871$ & $1.633872$ & $3513.18$ & $1015.39$ & $2.0\times10^{-6}$ \\
$0.21$ & $0.653347$ & $1.662318$ & $2978.82$ & $905.940$ & $2.1\times10^{-6}$ \\
$0.22$ & $0.673813$ & $1.687544$ & $2547.18$ & $813.212$ & $2.3\times10^{-6}$ \\
$0.23$ & $0.695058$ & $1.709610$ & $2195.01$ & $734.033$ & $2.4\times10^{-6}$ \\
$0.24$ & $0.716884$ & $1.728611$ & $1905.05$ & $665.938$ & $2.6\times10^{-6}$ \\
$0.25$ & $0.739105$ & $1.744664$ & $1664.30$ & $606.991$ & $2.4\times10^{-6}$ \\
$0.26$ & $0.761547$ & $1.757902$ & $1462.86$ & $555.654$ & $2.5\times10^{-6}$ \\
$0.27$ & $0.784055$ & $1.768465$ & $1293.10$ & $510.692$ & $2.6\times10^{-6}$ \\
\end{longtable}

\subsection*{The (simple) shooting method ($M=0$)}
{\footnotesize
\begin{longtable}{rcccccccc}
\caption{
    \footnotesize
3D, the (standard) shooting method; $m=1$, $\coupling=1$, $M=0$.
We recall that we compute for several values~$s$ the solution $(v_s, u_s, H_s)$ to the system~\eqref{system-dkg-r-2} for the parameters $\coupling=1$, $\omega_0=1$, and $H(0)=s$, by performing a (simple) shooting on $v(0)$. We then compute the effective mass $m_s$, from which we rescale the obtained solution into the solitary wave $(v, u, h)$ with associated quantities $Q$, $V$, $E$ etc. We refer for details to the subsection ``The (standard) shooting method'' of Section~\ref{sect-dkg-numerics-massless} in the main manuscript, and in particular to~\eqref{def-H}, and \eqref{def-ms}--\eqref{def-new-omega}.
}
\label{supp-table-3D-SimpleShootingMethod-M=0} \\
\toprule
$s$ & $v_s(0)$ & $m_s$ & $\omega$ & $v(0)$ & $h(0)$ & $Q$ & $E$ & $|\varepsilon|/(\omega Q)$\\
\midrule
\endfirsthead 

\multicolumn{2}{l}{\tablename\ \thetable\ } & \multicolumn{4}{r}{\emph{Continued from previous page}} \\ 
\hline
$s$ & $v_s(0)$ & $m_s$ & $\omega$ & $v(0)$ & $h(0)$ & $Q$ & $E$ & $|\varepsilon|/(\omega Q)$\\
\midrule
\endhead 

\cdashline{1-9}
\multicolumn{9}{r}{\emph{Continued on next page}} \\ 
\endfoot

\hline
\multicolumn{9}{r}{\emph{End of \tablename\ \thetable}} \\ 
\endlastfoot

$-0.99999$ & $0.000015$ & $1.000011$ & $0.999989$ & $0.000015$ & $0.000020$ & $0.10168$ & $0.10168$ & $4.6\times10^{-15}$ \\
$-0.9999$ & $0.000153$ & $1.000107$ & $0.999893$ & $0.000153$ & $0.000206$ & $0.32156$ & $0.32155$ & $6.8\times10^{-15}$ \\
$-0.999$ & $0.001539$ & $1.001066$ & $0.998935$ & $0.001536$ & $0.002063$ & $1.01731$ & $1.01695$ & $6.7\times10^{-16}$ \\
$-0.998$ & $0.003078$ & $1.002132$ & $0.997872$ & $0.003068$ & $0.004123$ & $1.43939$ & $1.43836$ & $1.3\times10^{-13}$ \\
$-0.997$ & $0.004617$ & $1.003198$ & $0.996812$ & $0.004595$ & $0.006178$ & $1.76372$ & $1.76184$ & $1.8\times10^{-13}$ \\
$-0.996$ & $0.006156$ & $1.004265$ & $0.995753$ & $0.006116$ & $0.008229$ & $2.03754$ & $2.03465$ & $2.3\times10^{-13}$ \\
$-0.995$ & $0.007694$ & $1.005331$ & $0.994697$ & $0.007633$ & $0.010276$ & $2.27913$ & $2.27508$ & $2.9\times10^{-13}$ \\
$-0.994$ & $0.009232$ & $1.006398$ & $0.993642$ & $0.009144$ & $0.012319$ & $2.49785$ & $2.49253$ & $9.2\times10^{-13}$ \\
\end{longtable}
}

\clearpage
\subsection*{The iterative method}

\begin{table}[!ht]
\begin{center}
\caption{\footnotesize
3D, the iterative shooting method; $m=1$, $\coupling=1$, and $M=1/2$.
}
\begin{tabular}{ccccccc} \toprule
$\omega$ & step & $v(0)$ & $h(0)$ & $Q$ & $E$ & $|\varepsilon|/(\omega Q)$\\
\midrule
\multirow{7}{*}{$0.2$} & $\#\,00$ & $1.202271$ &  & $10695.3$ & $2934.38$ & $1.4\times10^{-2}$ \\
 & $\#\,01$ & $0.753295$ & $1.848327$ & $9036.84$ & $2502.91$ & $1.5\times10^{-1}$ \\
 & $\#\,02$ & $0.506860$ & $1.465729$ & $1621.23$ & $478.477$ & $2.4\times10^{-1}$ \\
 & $\#\,04$ & $0.557267$ & $1.537437$ & $2279.24$ & $663.981$ & $1.1\times10^{-1}$ \\
 & $\#\,08$ & $0.608139$ & $1.602211$ & $3054.81$ & $885.105$ & $2.7\times10^{-2}$ \\
 & $\#\,16$ & $0.628682$ & $1.639337$ & $3702.92$ & $1055.50$ & $2.8\times10^{-2}$ \\
 & $\#\,20$ & $0.628682$ & $1.633640$ & $3528.42$ & $1015.67$ & $8.7\times10^{-3}$ \\
\bottomrule
\end{tabular}
\label{supp-table-3D-IterativeMethod-M=0.50}
\end{center}
\end{table}

\begin{table}[!ht]
\begin{center}
\caption{\footnotesize
3D, the iterative shooting method; $m=1$, $\coupling=1$, and $M=1/4$.
}
\begin{tabular}{ccccccc} \toprule
$\omega$ & step & $v(0)$ & $h(0)$ & $Q$ & $E$ & $|\varepsilon|/(\omega Q)$\\
\midrule
\multirow{7}{*}{$0.2$} & $\#\,00$ & $1.202408$ &  & $10853.2$ & $2983.48$ & $1.5\times10^{-4}$ \\
 & $\#\,01$ & $0.496822$ & $1.771768$ & $2435.84$ & $718.568$ & $1.1\times10^{-1}$ \\
 & $\#\,02$ & $0.379338$ & $1.598498$ & $907.126$ & $279.347$ & $6.1\times10^{-2}$ \\
 & $\#\,04$ & $0.419544$ & $1.643344$ & $1204.50$ & $369.305$ & $7.1\times10^{-4}$ \\
 & $\#\,08$ & $0.426879$ & $1.651741$ & $1268.44$ & $388.432$ & $8.3\times10^{-3}$ \\
 & $\#\,16$ & $0.427060$ & $1.651950$ & $1269.69$ & $388.833$ & $9.1\times10^{-3}$ \\
 & $\#\,20$ & $0.427060$ & $1.651950$ & $1269.69$ & $388.833$ & $9.1\times10^{-3}$ \\
\bottomrule
\end{tabular}
\label{supp-table-3D-IterativeMethod-M=0.10}
\end{center}
\end{table}

\begin{table}[!ht]
\begin{center}
\caption{\footnotesize
3D, the iterative shooting method; $m=1$, $\coupling=1$, and $M=1/10$.
}
\begin{tabular}{ccccccc} \toprule
$\omega$ & step & $v(0)$ & $h(0)$ & $Q$ & $E$ & $|\varepsilon|/(\omega Q)$\\
\midrule
\multirow{7}{*}{$0.2$} & $\#\,00$ & $1.202271$ &  & $10695.3$ & $2934.38$ & $1.4\times10^{-2}$ \\
 & $\#\,01$ & $0.425067$ & $1.825548$ & $2049.23$ & $599.433$ & $1.8\times10^{-1}$ \\
 & $\#\,02$ & $0.307532$ & $1.671609$ & $0.0$ & $0.0$ & $\infty$ \\
 & $\#\,04$ & $0.336045$ & $1.705893$ & $919.499$ & $276.941$ & $1.5\times10^{-1}$ \\
 & $\#\,08$ & $0.342288$ & $1.726219$ & $0.0$ & $0.0$ & $\infty$ \\
 & $\#\,16$ & $0.344343$ & $1.712332$ & $0.0$ & $0.0$ & $\infty$ \\
 & $\#\,20$ & $0.344343$ & $1.726789$ & $1077.73$ & $317.512$ & $1.8\times10^{-1}$ \\
\bottomrule
\end{tabular}
\label{supp-table-3D-IterativeMethod-M=0.25}
\end{center}
\end{table}

\begin{table}[!ht]
\begin{center}
\caption{\footnotesize
3D, the iterative shooting method; $m=1$, $\coupling=1$, and $M=0$.
}
\begin{tabular}{ccccccc} \toprule
$\omega$ & step & $v(0)$ & $h(0)$ & $Q$ & $E$ & $|\varepsilon|/(\omega Q)$\\
\midrule
\multirow{7}{*}{$0.4$} & $\#\,00$ & $1.343712$ &  & $872.035$ & $470.086$ & $4.7\times10^{-4}$ \\
 & $\#\,01$ & $0.637020$ & $1.464074$ & $125.838$ & $76.5880$ & $1.1\times10^{-1}$ \\
 & $\#\,02$ & $0.626092$ & $1.388163$ & $119.845$ & $74.0241$ & $9.5\times10^{-2}$ \\
 & $\#\,04$ & $0.614813$ & $1.368097$ & $115.447$ & $71.5263$ & $9.1\times10^{-2}$ \\
 & $\#\,08$ & $0.614343$ & $1.367374$ & $115.262$ & $71.4179$ & $9.1\times10^{-2}$ \\
 & $\#\,16$ & $0.614365$ & $1.367454$ & $115.266$ & $71.4181$ & $9.2\times10^{-2}$ \\
 & $\#\,20$ & $0.614432$ & $1.367526$ & $115.295$ & $71.4362$ & $9.1\times10^{-2}$ \\
\bottomrule
\end{tabular}
\label{supp-table-3D-IterativeMethod-M=0}
\end{center}
\end{table}

\clearpage
\section*{1D case}
\subsection*{The nested shooting method}

\begin{longtable}{lcccccccc}
\caption{
    \footnotesize
    1D, the nested shooting method; $m=1$, $\coupling=1$, and $M=1$.
}
\label{supp-table-1D-NestedshootingMethod-M=1} \\
\toprule
$\omega$ & $v(0)$ & $h(0)$ & $Q$ & $E$ & $|\varepsilon|/(\omega Q)$ \\
\midrule
\endfirsthead 

\multicolumn{2}{l}{\tablename\ \thetable\ } & \multicolumn{4}{r}{\emph{Continued from previous page}} \\ 
\hline
$\omega$ & $v(0)$ & $h(0)$ & $Q$ & $E$ & $|\varepsilon|/(\omega Q)$ \\
\midrule
\endhead 

\cdashline{1-6}
\multicolumn{6}{r}{\emph{Continued on next page}} \\ 
\endfoot

\hline
\multicolumn{6}{r}{\emph{End of \tablename\ \thetable}} \\ 
\endlastfoot

$0.00001$ & $1.414216$ & $1.994867$ & $2.0\times10^{5}$ & $25.2756$ & $2.4\times10^{-4}$ \\
$0.0001$ & $1.414189$ & $1.983848$ & $20001.3$ & $20.6709$ & $2.0\times10^{-7}$ \\
$0.001$ & $1.413937$ & $1.949869$ & $2001.35$ & $16.0657$ & $1.0\times10^{-7}$ \\
$0.005$ & $1.412631$ & $1.890516$ & $401.380$ & $12.8469$ & $5.4\times10^{-8}$ \\
$0.01$ & $1.410786$ & $1.847015$ & $201.402$ & $11.4608$ & $9.7\times10^{-8}$ \\
$0.02$ & $1.406687$ & $1.786049$ & $101.427$ & $10.0748$ & $3.3\times10^{-8}$ \\
$0.03$ & $1.402209$ & $1.739264$ & $68.1050$ & $9.26420$ & $3.8\times10^{-7}$ \\
$0.04$ & $1.397441$ & $1.699592$ & $51.4425$ & $8.68898$ & $4.3\times10^{-7}$ \\
\end{longtable}

\begin{longtable}{lcccccccc}
\caption{
    \footnotesize
    1D, the nested shooting method; $m=1$, $\coupling=1$, and $M=1/2$.
}
\label{supp-table-1D-NestedshootingMethod-M=0.5} \\
\toprule
$\omega$ & $v(0)$ & $h(0)$ & $Q$ & $E$ & $|\varepsilon|/(\omega Q)$ \\
\midrule
\endfirsthead 

\multicolumn{2}{l}{\tablename\ \thetable\ } & \multicolumn{4}{r}{\emph{Continued from previous page}} \\ 
\hline
$\omega$ & $v(0)$ & $h(0)$ & $Q$ & $E$ & $|\varepsilon|/(\omega Q)$ \\
\midrule
\endhead 

\cdashline{1-6}
\multicolumn{6}{r}{\emph{Continued on next page}} \\ 
\endfoot

\hline
\multicolumn{6}{r}{\emph{End of \tablename\ \thetable}} \\ 
\endlastfoot

$0.00001$ & $0.707424$ & $1.941454$ & $50045.5$ & $6.69306$ & $5.4\times10^{-7}$ \\
$0.0001$ & $0.708110$ & $1.897054$ & $5014.77$ & $5.54080$ & $9.5\times10^{-8}$ \\
$0.001$ & $0.710211$ & $1.818729$ & $505.029$ & $4.38642$ & $7.0\times10^{-8}$ \\
$0.005$ & $0.713563$ & $1.728314$ & $102.548$ & $3.57617$ & $7.4\times10^{-7}$ \\
$0.01$ & $0.715573$ & $1.674485$ & $51.9707$ & $3.22552$ & $9.8\times10^{-9}$ \\
$0.02$ & $0.717517$ & $1.607059$ & $26.5715$ & $2.87330$ & $2.7\times10^{-7}$ \\
$0.03$ & $0.718191$ & $1.559035$ & $18.0647$ & $2.66633$ & $3.1\times10^{-8}$ \\
$0.04$ & $0.718145$ & $1.519924$ & $13.7951$ & $2.51892$ & $5.1\times10^{-8}$ \\
\end{longtable}

\clearpage
\subsection*{The iterative method}

\begin{table}[!ht]
\begin{center}
\caption{\footnotesize
1D, the iterative shooting method; $m=1$, $\coupling=1$, and $M=1/2$.
}
\begin{tabular}{ccccccc} \toprule
$\omega$ & step & $v(0)$ & $h(0)$ & $Q$ & $E$ & $|\varepsilon|/(\omega Q)$\\
\midrule
\multirow{7}{*}{$0.2$} & $\#\,00$ & $1.264976$ &  & $9.79592$ & $4.58271$ & $5.7\times10^{-4}$ \\
 & $\#\,01$ & $0.809408$ & $1.258211$ & $0.0$ & $0.0$ & $\infty$ \\
 & $\#\,02$ & $0.712179$ & $1.183841$ & $3.57891$ & $1.77919$ & $2.1\times10^{-1}$ \\
 & $\#\,04$ & $0.685492$ & $1.161416$ & $3.35792$ & $1.67680$ & $1.2\times10^{-2}$ \\
 & $\#\,08$ & $0.683374$ & $1.159553$ & $3.34093$ & $1.66896$ & $5.0\times10^{-3}$ \\
 & $\#\,16$ & $0.683356$ & $1.159534$ & $3.34080$ & $1.66890$ & $5.3\times10^{-3}$ \\
 & $\#\,20$ & $0.683352$ & $1.159528$ & $3.34076$ & $1.66889$ & $5.4\times10^{-3}$ \\
\bottomrule
\end{tabular}
\label{supp-table-1D-IterativeMethod-M=0.50}
\end{center}
\end{table}